\documentclass[12pt, draftclsnofoot, onecolumn]{IEEEtran}

\usepackage{setspace}
\usepackage{cite}
\usepackage{amsthm}
\usepackage{epsfig}
\usepackage{amsmath}
\usepackage{amssymb}
\usepackage{enumitem}
\usepackage{mathtools}
\usepackage{amsfonts}
\usepackage{color}
\usepackage{optidef}
\usepackage{lettrine}
\usepackage{graphicx}
\usepackage{multicol,flushend,color}
\usepackage{wasysym}
\usepackage{float}
\usepackage{algorithm}
\usepackage[noend]{algpseudocode}
\usepackage{xcolor}
\usepackage{cuted}
\usepackage{stfloats} 
\usepackage{caption}

\DeclarePairedDelimiter{\ceil}{\lceil}{\rceil}
\DeclarePairedDelimiter{\floor}{\lfloor}{\rfloor}
\newtheorem{proposition}{Proposition}
\newtheorem{lemma}{Lemma}

\def \treq {\stackrel{\tiny \Delta}{=}}

\newcommand{\E}{\ensuremath{\mathbb E}}

\newcommand{\src}{$\mathcal{S}$}

\begin{document}

\title{On Learning-Assisted Content-Based Secure Image
	Transmission for Delay-Aware Systems with
	Randomly-Distributed Eavesdroppers \\ 
--- Extended Version ---}
\author{
}
\markboth{}
{}
\author{\IEEEauthorblockN
{Mehdi Letafati, \IEEEmembership{Student Member, IEEE,}
			Hamid Behroozi, \IEEEmembership{Member, IEEE,}
		 Babak Hossein Khalaj, 
			\IEEEmembership{Senior Member, IEEE,}  and
		Eduard A. Jorswieck,
		   	\IEEEmembership{Fellow, IEEE,}
}
\textsuperscript{}\thanks{ 
	 M. Letafati, H. Behroozi, and B. H. Khalaj are with the Department of Electrical Engineering, Sharif University of Technology, Tehran 1365-11155, Iran (e-mails: mletafati@ee.sharif.edu; behroozi@sharif.edu;  khalaj@sharif.edu). 
	 E. A. Jorswieck is with the Department of Information Theory and
	 Communication Systems, TU Braunschweig, Germany (e-mail: e.jorswieck@tu-bs.de). 
}}
\maketitle

\begin{abstract}
In this paper, a learning-aided content-based wireless image transmission scheme is proposed, where a multi-antenna-aided source wishes to securely deliver an image to a legitimate destination 
in the presence of randomly distributed eavesdroppers (Eves).   
We take into account the fact that not all regions of an image have the same {importance} from the security perspective. 
Hence, we propose a transmission scheme, where the 
source employs a hybrid method to realize both the  error-free data delivery of public regions---containing  less-important pixels; and an artificial noise (AN)-aided transmission scheme  {to provide security} for the regions containing large amount of information.     
Moreover, in order to reinforce system's security, fountain-based packet delivery is  adopted: First,  
the source node  encodes image packets into fountain-like packets prior to sending them over the air. 
{The secrecy of our proposed scheme will be achieved if the legitimate destination correctly receives the entire image source packets, while conforming to the latency limits of the system, before Eves can obtain the important regions.  
Accordingly,} the secrecy performance of our scheme  is characterized by deriving the closed-form expression for {the quality-of-security (QoSec) violation probability.}  
Moreover, our proposed wireless image delivery scheme  leverages the deep neural network (DNN) and learns to maintain optimized transmission parameters, while achieving a low QoSec violation probability. 
Simulation results are provided with some useful engineering insights which  illustrate that our proposed  learning-assisted scheme outperforms the state-of-the-arts by achieving  considerable gains in terms of security and the delay requirement. 
\end{abstract}


\IEEEpeerreviewmaketitle
\section{Introduction} 
\IEEEPARstart{T}{he} modern era of wireless communication  is facing  a remarkable  evolution that has led to the development of 
 beyond-the-fifth-generation (B5G) mobile technologies, 
which assert supporting an overabundance of  contemporary services
including the Internet-of-Things (IoT), autonomous driving, and Augmented/Virtual Reality (AR/VR) \cite{VR}. 
Meanwhile, comprehensive studies are being conducted on the sixth generation (6G) of wireless cellular systems to look for its potential development  challenges \cite{6G}.    

Despite the genuinely-designed networking  protocols of 5G and B5G communication systems,  security challenges are still 
witnessed  as open issues that have not been thoroughly addressed. 
Simultaneously, on the roadmap towards 6G, new security challenges can be identified due to the considerable revisions of key communication parameters and operational entities, e.g., the end-to-end  tolerable latency;  the deployment of long-lasting IoT devices with security requirement;  the wide-range usage of different RF technologies; and, the quantum computers {\cite{6G-PLS,6G}}.    

\subsection{Motivation}
Although a great amount of innovative studies have been carried out on the security of the 5G core network \cite{net-slic},  
 the wireless edge of 5G, B5G, and also 6G systems is vastly encountered with ever-rising  attacks, e.g., jamming, eavesdropping, and traffic analysis, which can be deployed at a low price,  using low-cost software defined radios. 
 More precisely, the intrinsic broadcast nature of wireless medium  has made the wireless networks  susceptible to security and  privacy risks  \cite{IoT1}.  
    
 In this regard, physical layer security (PLS) solutions, which leverage intrinsic properties of the radio channel, can be incorporated  thanks to their inherent capability of being adapted  to the communication medium \cite{IoT1,IoT2,6G-PLS}. 
 The main idea of PLS is to make the legitimate link  better than
 the unauthorized links of potential adversaries.    
 Remarkably, a great number of new attributes of future networks, such as low-latency control systems and sensor fusions, require only localized communications, where the core network is not engaged. 
  Hence, utilizing  PLS approaches can provide a desirable level of agile security for such scenarios.    
In addition to rendering  flexible secure frameworks, PLS can achieve  information-theoretic security guarantees via utilizing lightweight mechanisms \cite{IoTJ}.

Remarkably,  owing to the ability of machine learning (ML)-based  techniques to acquire and learn  complex features of data, utilizing ML can help PLS solutions propel the secrecy performance of wireless systems to a state that reliable  context-aware security can be achieved with a low computational complexity \cite{6G-PLS}.  
For instance, the secrecy of wireless communications can be enhanced via utilizing deep neural networks (DNNs) to intelligently maintain optimal  transmission parameters {\cite{learning-based,  Eduard-AE}}.    

Recently, wireless image transmission,  as a practical scenario of reliable data delivery, is receiving much attention from both academia and industry,  due to its applications in newly-emerged services such as AR/VR or surveillance systems \cite{Im-retriev-deniz}.   
In this regard, if the transmitted  images are not rigorously secured, the performance of such wireless services are no longer reliable.    Moreover, from the quality of service (QoS) point of view, the delay  imposed to the imaging systems due to the transmission of large-sized  images should be handled properly.  
Conventional wireless image transmission schemes were realized through computationally expensive  cryptographic methods or digital watermarking techniques \cite{camouf}, which are easily  breakable with the increasing growth in the computational capabilities of modern computers, e.g.,  quantum computing technology.  
Therefore, the PLS-based approaches  together with ML techniques can be leveraged to provide flexible delay-  and content-aware security solutions for  wireless image delivery.     

\subsection{Related Works}
A large number of researches in the PLS domain has been carried out to  study the so-called secrecy capacity (SC), which is defined as the maximum data rate that can be reliably and confidentially transmitted, e.g., in some case,  the difference between the capacities of the main and wiretap
	channels 
	to provide perfect secrecy {\cite{3hop, twc, Leti-GC}}.  
However, the SC-inspired transmission design dictates a much lower acceptable data rate, which may result in large delivery delays.  
Despite the growing  demand for higher capacities, the emergence of the IoT-enabled services is imposing a  paradigm shift from the capacity-centric  mobile  services  towards reliable delay-aware and content-based services. 
Accordingly, in a practical scenario of wireless image  transmissions,    the perfect secrecy is not realizable 
{\cite{basic-medical, sun2020,sun-Ind-relay, sun-prediction}}.   
In other words, from the view-point of an eavesdropping adversary, the  original image cannot be retrieved even if a small number of packets are correctly received. 

In the area of PLS-based data delivery, the authors in \cite{basic-medical} examined  an adaptive allocation of physical resources in a fountain coding (FC)-aided point-to-point (P2P) system of delivering medical images.    
A similar single-input single-output (SISO) FC-assisted wireless network was examined in \cite{sun2020} and \cite{sun-prediction} with  more emphasis on  the image content and the channel imperfections, respectively.   
The FC as a low-complexity technique has been shown to be able to improve   packet delivery security at the physical layer (PHY) {\cite{FC1}}.  
In fountain-coded transmissions, the original data is first divided into a series of, e.g.,  $N$, source packets, which are then linearly combined to obtain 
encoded packets.  
The source message can be entirely recovered at receiver, iff at least $N$ independent FC packets are successfully received \cite{FC1}. 
Remarkably, despite addressing the delay limits of a practical system in \cite{basic-medical,sun2020,sun-prediction}, the researchers did not provide a \emph{closed-form expression} for their performance metric.  
 The FC-aided packet delivery was  applied to cooperative networks in  \cite{sun-Ind-relay}, by deploying more complex techniques in addition to FC, i.e.,  cooperative jamming and relay selection. 
 However, the researchers' proposed scheme in \cite{sun-Ind-relay} did not take  the content of transmitted data into account, yielding an unconscious packet delivery scheme. 
 Moreover, a fixed rate transmission was considered in their work regardless of wisely considering the link quality and the content of  image. 
Notably, all of the studies  mentioned above,  considered  only one single-antenna eavesdropper (Eve)  in their model, where some partial information regarding the wiretap channel is available.   

\subsection{Contributions} 
In this paper, different from {\cite{basic-medical,sun-Ind-relay,sun2020,FC1, sun-prediction}},
we  investigate a 
multi-input single-output (MISO) FC-aided wireless system, where  a multi-antenna source node aims  to send an image to a legitimate destination, while conforming to the acceptable  delay limits of  delivering the file.     
In addition to the legitimate entities, different from    {\cite{3hop,twc,IoTJ, Leti-GC,basic-medical,sun2020,FC1,sun-Ind-relay,sun-prediction}}, in our proposed system model we assume    \textit{multiple totally passive} randomly-located  Eves which are trying to wiretap the wireless image  delivery.  
The eavesdroppers in our system model are assumed to be distributed according to homogeneous Poisson point process (PPP) \cite{Stoch-Geom}.    
We also consider both cases of non-colluding eavesdropping (NCE) and colluding eavesdropping (CE) scenarios throughout our analysis and results.  
 Moreover, it is considered in our system model that we have an uncertain (imperfect) estimation regarding the legitimate link, {while there is no channel state information (CSI) available at the legitimate source about the eavesdropping links.}

We address the content of the image through the process of secure wireless transmission 
in the sense that not all regions of an  image have equal importancy levels from the 
{content-centric view}.  
Accordingly, we consider the image to be partitioned  into two segments, i.e., the region of interest (RoI) and the background (BG), where the RoI source packets of an image are likely to contain important diagnostic information regarding the content of an image. 
Hence, the RoI packets will require stringent security and reliability specifications,  while the BG packets do not provide useful data if being revealed to the passive Eves.   
Remarkably, a  content-based secure   wireless image delivery  is not addressed in many recently-published works in the area of wireless image transmission, e.g.,   
{\cite{Im-retriev-deniz,DJSCC-Deniz,AE-Deniz}}.  
In other words, they all consider a unified end-to-end approach for sending an image over the air, regardless of paying attention to the content/region of the image being transmitted at each transmission slot (TS).    
In contrast, we propose a hybrid transmission policy according to the type of the source packets scheduled to be transmitted. 
When sending RoI-related 
 packets, our proposed system model deploys  maximum ratio transmission (MRT) beamforming together with artificial noise (AN) injection, as a practical and low cost PLS technique \cite{ANI-practical},  to achieve a physical enhancement for the legitimate link, while degrading the eavesdropping ability of potential Eves.  
For the transmission of BG packets, we simply deploy MRT in order to save power consumption   and improve the packet accumulation ability of legitimate receiver, since the BG regions require less stringent security.   
We also stress that in order to reinforce the system's capability for reliable  transmission of image source file, our proposed scheme also exploits FC  and encodes the original source packets into fountain-like packets prior to  MRT deployment.        

In this paper, different from  the  previous works of   {\cite{Leti-GC,3hop,twc,IoTJ}}, and \cite{AE-Deniz},   we take the delay limits of a practical image transmission system into account  
via utilizing a hybrid multi-packet transmission design.  
We then derive an exact closed-form expression, which was not obtained in {\cite{basic-medical,sun2020,sun-prediction}}, for the {quality-of-security (QoSec) violating probability (QVP)} as our performance metric to addresses the  effect of delay limits on the image delivery.     
 {The QVP‌ metric characterizes the performance of our system to correctly recover the entire image at destination before Eves can obtain the RoI-related important  regions, while conforming to the delivery delay limits.}

In our proposed model, we further formulate an optimization problem opted for reducing the QVP by maintaining optimal transmission parameters.  
Due to the non-linear characteristics of our objective and constraints, finding optimal solutions through conventional algorithms may result in low efficiency and long delay effect  in a practical deployment. 
Therefore, we 
utilize a DNN composed of dense layers as our learning-based approach  to provide an efficient and flexible  method of finding optimal transmission parameters. 
In other words, we exploit the potential of DNN in intelligently  adjusting our image transmission scheme to minimize the QVP performance metric, which was not considered in {\cite{3hop,twc,Leti-GC,basic-medical,sun2020,sun-prediction,sun-Ind-relay}}.   
By leveraging the DNN,  we perform an offline phase of data set generation for numerous realizations of the wireless medium. 
 Our developed DNN then  learns an optimum mapping  between the generated  realizations and the associated  optimal policies.  
  Afterwards, our trained DNN 
  provides {approximate  performance}  for  new realizations of the network via  a significantly-low computational complexity that actually  corresponds to performing  a forward propagation through  the trained DNN online.  
Finally, our simulation results clearly illustrate the performance of our proposed learning-based delay-aware scheme compared with the state-of-the-art benchmarks for   image transmission.  


The remainder of this paper is orchestrated as follows. In Section II, our proposed system model together with the main assumptions and intuitions are  provided. 
Section III  proposes our transmission policy for secure wireless image delivery.   
In Section IV, we investigate the performance of our proposed scheme  through deriving closed-form expression of the QVP.  
The details on how to obtain optimal transmission parameters for the ML aspect of our proposed scheme are studied in Section V.   ‌
In Section VI, extensive simulation results are provided to verify the analytical  results. We also compare our  learning-based scheme with different  benchmarks and baselines.    
Section VII concludes the paper and presents useful insights for further  research directions.

{\textit{Notations:} We denote the transpose, the conjugate transpose, and $\ell^2$ norm of a vector by $(\cdot)^\mathsf{T}$, $(\cdot)^\dagger$, and $||\cdot||$, respectively. 
Moreover, $|\cdot|$ represents
the absolute value of a scalar variable and 
 the cardinality of a set.    
Vectors are represented by bold lowercase letters, while matrices are written  as bold uppercase symbols.    
$\mathcal{CN}(\mu,\sigma^2)$, $\text{Exp}(\lambda)$, and $\mathcal{G}(N,\lambda)$  represent a complex random variable (RV) with mean $\mu$ and variance $\sigma^2$, exponential
distribution with parameter $\lambda$, and gamma distribution with
parameters $N$ and $\lambda$, respectively.  
The expected value, the probability density function (pdf) (also  the probability mass function (PMF) for a discrete RV), the cumulative distribution function (CDF), and the complementary cumulative distribution function
(CCDF) of RV ${X}$ are denoted by $\E[{X}]$, $f_X(x)$,  $F_X(x)$, and $\overline{F}_X(x)$, respectively. 
Moreover, the probability of an event $A$ is denoted by $\mathsf{Pr}(A)$.
$W_{\nu,\mu}(x)$, $\Gamma(x)$, $B(m,n)$,  and  $\gamma(\alpha,\beta)$ 
are the Whittaker function {\cite[Eq. (9.22)]{integ}}, 
the Gamma function {\cite[Eq. (8.31)]{integ}}, the Beta function {\cite[Eq. (8.384)] {integ}}, and the lower incomplete Gamma function {\cite[Eq. (8.350)] {integ}}, 
respectively.}


\begin{figure}
	\centering
	\begin{minipage}{0.45\textwidth}
		\centering
		\includegraphics
			[width=3.3in,height=2.5in,trim={0.5in 0.0in 0 0.5in},clip]{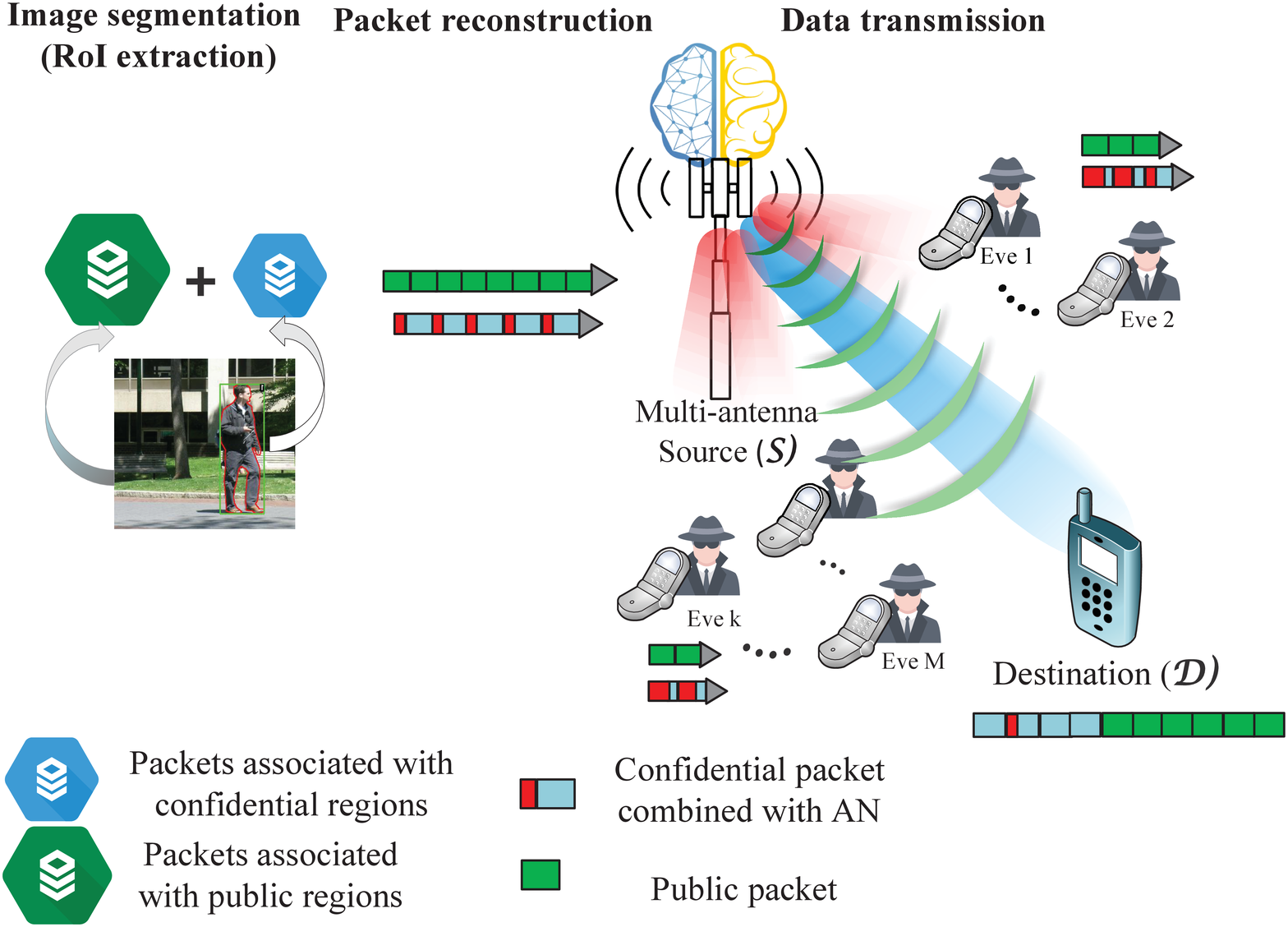}
		\caption{\small Proposed Learning-assisted System Model for Packet-based Wireless Image Transmission.}
		\label{fig:SysModel}\vspace{-7mm}
	\end{minipage}\hfill
	\begin{minipage}{0.45\textwidth}
		\includegraphics
		[width=3.2in,height=1.85in,
		trim={0.3in 0.0in 0 0.1in},clip]{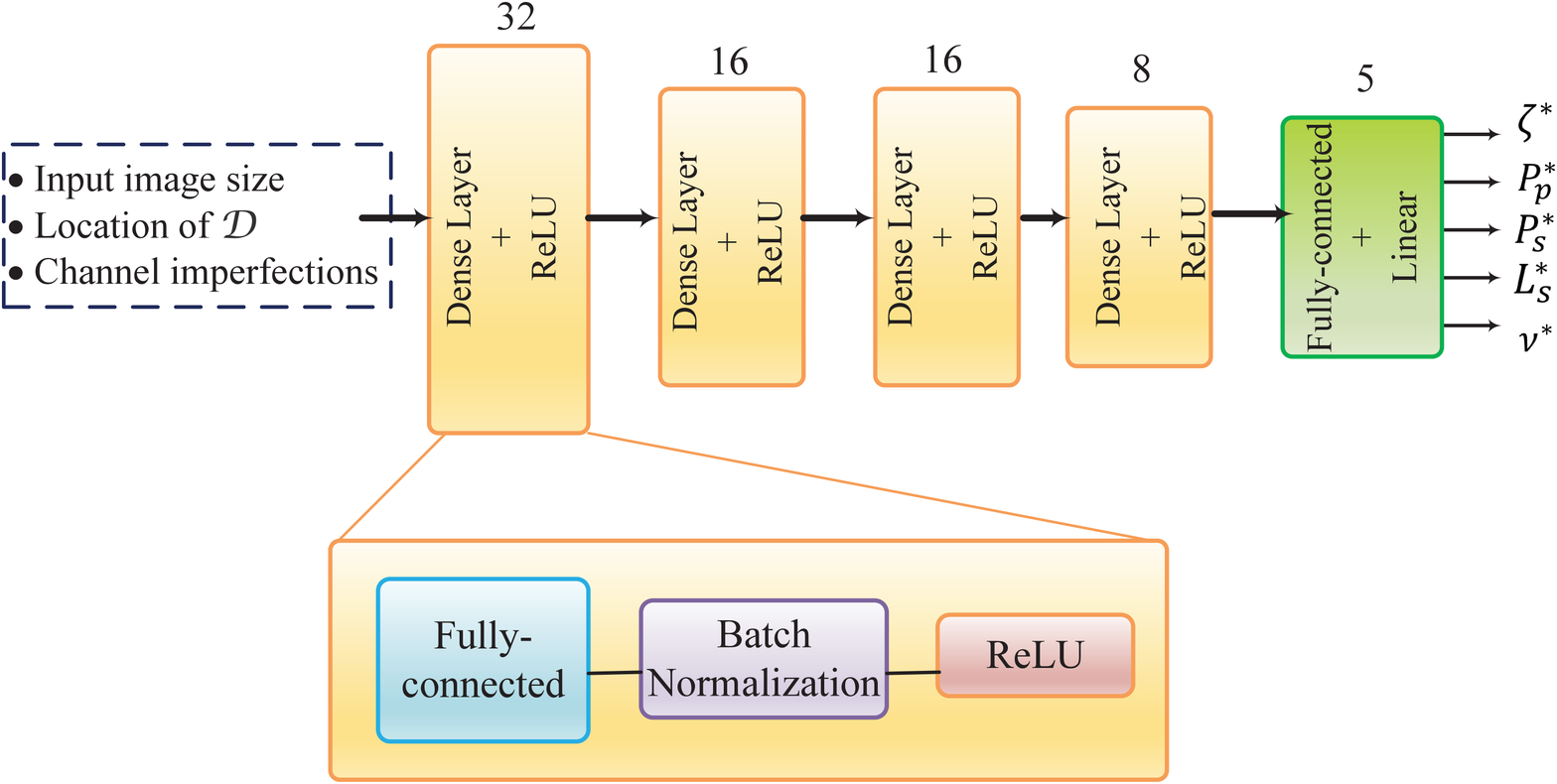}
		\caption{\small Our proposed DNN for content-aware wireless image transmission. The number of neurons in each layer is shown at the top of that layer. }
		\label{fig:DNN-proposed}
		\vspace{-7mm}
	\end{minipage}
\end{figure}

\section{System Model} 	
Consider a MISO ‌wireless image transmission system, where the source node $\mathcal{S}$, equipped with $n_\mathsf{T}$ transmit antennas,  delivers an image file to a single-antenna legitimate destination node $\mathcal{D}$, as depicted in Fig. \ref{fig:SysModel}.
Due to broadcast properties of wireless medium, the transmission is concurrently overheard by 
randomly distributed single-antenna \emph{totally-passive}  eavesdroppers (Eves), where the locations of Eves are modeled as a homogeneous  PPP $\Phi_E$ of density $\lambda_E$ on a $2$-D plane with the $i$-th Eve at distance $r_i$ 
($i\in\Phi_E$)  
from $\mathcal{S}$.  
We consider two eavesdropping strategies for the existing Eves in our proposed model, i.e., the NCE case, where the wiretap capacity is determined by the maximum capacity among all $\mathcal{S}$-to-Eves links;  and the CE case, in which Eves  cooperate with each other to improving their ability of wiretapping image delivery. 
Without loss of generality, we presume $\mathcal{S}$ to be placed at the origin and $\mathcal{D}$ at a known  position with  distance $r_D$ from $\mathcal{S}$. In other words, the source node is fully aware of the legitimate destination's position, while  presumes a probabilistic model for the location of passive Eves.     

To deliver the image source file, $\mathcal{S}$ first splits the file into $N$ source packets denoted by $\{\pi_1, \pi_2, \cdots,  \pi_N\}$,  each containing the same number $b$ of bits. 
The source packets will further be transformed into a potentially unlimited number of fountain-based coded packets.  
Each coded packet is obtained by XORing a subset of the source packets. More details about fountain-based packet reconstruction will be mentioned later. 
During each TS, one or several fountain packets  form  the transmit frame with duration $T$ at PHY which is then sent out over the wireless channel. The available bandwidth for each slot is $B$ Hz. 
At the end of any time slot, the legitimate receiver $\mathcal{D}$, as well as Eves, performs channel decoding on the received signals.  
Successfully-received packets are acknowledged by sending an ACK packet through  a feedback channel from $\mathcal{D}$ to $\mathcal{S}$.
The feedback signals, do not contain any information regarding the original image. 
The transmission is terminated  when $\mathcal{D}$ gathers $N$ independent fountain-based coded packets to decode the entire original source file or the delay limit arrives.    
To guarantee image delivery security, the number of independent fountain-encoded  packets collected at Eves should be less than $N$ at the time of delivery termination; In addition, the Eves should not obtain any useful information corresponding to sensitive (confidential) regions of an image, e.g., the details of someone's face in a {secrecy-assuring} transmission.\footnote{The operational meaning of this measure is comprehensively studied  in Section \ref{Analysis} of our paper.}

In this paper,  we take into account the fact that not all regions of an image have the same {importance level  from the security/content perspective}.  
For instance, the  RoI packets tend to contain the highest portion of  information regarding the source  image; hence, they need a higher security and reliability requirement. 
On the contrary, the BG regions of an image, in most cases, do not contain  any useful information about the targeted image.   
Accordingly, we propose a hybrid transmission method for the delivery  of source packets  based on their type to ensure a reliable and secure transmission for the entire image file.      
More specifically, the source image is first  divided into two parts:  
the confidential packets, related to the RoI, which is subject to security constraints and tends to contain the highest portion 
of  information about the entire image; and the {public part}, namely  the background (BG), which can be sent over the air with less stringent security limits, since  revealing it to potential Eves does not leak any useful data.  
Notably, this segmentation can be realized via utilizing   threshold-based image saliency detection methods \cite{sun2020}.\footnote{
		Investigating different methods of image segmentation is out of the scope of this article.}
Mathematically speaking,  at $\mathcal{S}$'s side and after segmenting the source image,   extracted  packets  form two disjoint sets: $\Pi_\mathsf{s}$,  corresponding to packets with security constraints;  and  $\Pi_\mathsf{p}$, corresponding to  public BG packets,  respectively,  
where $\Pi_\mathsf{s} \cap \Pi_\mathsf{p} = \varnothing$,   $|\Pi_\mathsf{s}|=N_\mathrm{roi}$, and $|\Pi_\mathsf{p}|=N_\mathrm{bg}=N-N_\mathrm{roi}$. 
 $N_\mathrm{roi}$ and $N_\mathrm{bg}$ respectively denote the number of RoI and BG packets. 
The two sets of $|\Pi_\mathsf{s}|$ and $|\Pi_\mathsf{p}|$ are illustrated, respectively,  by blue and green stacks in Fig. \ref{fig:SysModel}.   

At each TS, $\mathcal{S}$ estimates its  legitimate link's condition; determines which kind of packet to be sent; and provides appropriate resource for transmission.   
Remarkably, a trained DNN is developed at ‌the source node which is opted for facilitating  the procedure of choosing the best transmission parameters. The details on how to train and utilize the DNN deployed at $\mathcal{S}$ will be discussed  later.    
Selected source packets are then encoded to form a transmitting frame.  
If the confidential packets are selected to form the frame, coded packets are further combined with AN, shown by red data packets in Fig. \ref{fig:SysModel}, to ensure   secrecy for the confidential packets $\{\pi \vert \pi \in \Pi_\mathsf{s}\}$.   
The 
frame is then sent over the wireless medium by exploiting MRT beamforming (together with the injected AN, if a confidential region is going to be  transmitted).    
At the receiving end, both $\mathcal{D}$ and Eves attempt to gather  sufficient number of packets to recover the source image.

All wireless channels are assumed to undergo 
quasi-static flat-fading with a large-scale path loss exponent $\eta>2$, where  	
the channel coefficients remain constant during each time slot and vary independently among different slots. Also, the channel coefficients of different links are independent of each other.  
The channel vector from $\mathcal{S}$  to a node $j$ with a distance $r_j$ is characterized as $\boldsymbol{h}_{sj}r^{-\frac{\eta}{2}}_j$, where $\boldsymbol{h}_{sj}\in \mathbb{C}^{n_\mathsf{T}\times1}$ denotes the small-scale fading vector, with independent and identically distributed (i.i.d.) elements ${h}_{sj,i} \sim \mathcal{CN}(0,1)$.
In order to obtain an estimation regarding the quality of the legitimate link at $\mathcal{S}$, the destination node sends pilot signals at the beginning of each TS, and $\hat{\boldsymbol{h}}_{sd}$ is obtained at $\mathcal{S}$ \cite{derrick-MIMO}.  
In this paper, we assume imperfect channel acquisition at legitimate nodes. Accordingly,  the fading vector of the exact main channel from  $\mathcal{S}$ to $\mathcal{D}$, denoted by $\boldsymbol{h}_{sd}$, is different from $\hat{\boldsymbol{h}}_{sd}$ due to 
imperfect channel reciprocity caused by the outdated  CSI \cite{derrick-MIMO}.  
Hence, the fading vector of main channel is modeled as
\begin{align}\label{channel-imp}
\boldsymbol{h}_{sd}=\rho\hat{\boldsymbol{h}}_{sd}+\sqrt{1-\rho^2}\hat{\boldsymbol{e}}_{sd}, 
\end{align}
where  $\hat{\boldsymbol{e}}_{sd}$ reflects the uncertain part of $\hat{\boldsymbol{h}}_{sd}$ 
with i.i.d. entries   $\hat{\boldsymbol{e}}_{sd,i} \sim \mathcal{CN}(0,1)$. Moreover, $\rho\in(0, 1)$ denotes the correlation coefficient between $\boldsymbol{h}_{sd}$  and $\hat{\boldsymbol{h}}_{sd}$, which   can easily be calculated from the knowledge of wireless medium parameters \cite{derrick-MIMO}.  
We also define $\hat{g}_{sd}\treq||\hat{\boldsymbol{h}}_{sd}||^2$ as the estimated channel gain of the legitimate $\mathcal{S}$-to-$\mathcal{D}$ link,  where the estimation is done by $\mathcal{S}$.
We also assume that the CSIs of $\mathcal{S}$-to-Eves are unknown.  
Without loss of generality, we assume that the noise $n_j$ at each node $j$ is statistically independent additive white Gaussian noise (AWGN) and follows $\mathcal{CN}(0,\sigma_n)$.

\section{Design of the Transmission Scheme}
In this section, we present a detailed description of our proposed image transmission scheme, including the strategies for fountain-encoding  of source packets, transmit frame reconstruction, and the signaling design of our proposed semi-adaptive transmission. 
Prior to image delivery, the source file is firstly divided into $N$  source packets corresponding to $N$ non-overlapping regions.
The extracted source packets are then partitioned   into two sets of secure (RoI-related) packets, $\Pi_\mathsf{s}$, and public ones, denoted by $\Pi_\mathsf{p}$.  

Without loss of generality, we consider the file delivery during arbitrary slots of $t-1$ and $t$. 
We also emphasize that the allowable delay limit for the image delivery  to take place is $D_{lim}$. Therefore, we have $t\leq D_{lim}$. 
In the $t$-th TS, in order to construct a transmitting data frame,  
$\mathcal{S}$ estimates the instantaneous  gain $\hat{{g}}_{sd}(t)$ 
to obtain an estimate about the quality of legitimate link.  
Hereafter, if the instantaneous gain $\hat{{g}}_{sd}(t)$ exceeds a threshold $\nu$, indicating a sufficiently good link quality, confidential  packets  corresponding to RoI regions are called to construct transmit frame.  Otherwise, the transmitting frame will comprise public packets \cite{basic-medical}.   
Next step is to establish the 
semi-adaptive transmission strategy. 
That is, if the public packets are selected to construct the frame, an adaptive number of packets (denoted by $L_p$) which are aimed to be transmitted during the $t$-th slot is derived. 
Calculation of $L_p$ is done in a way that  ensures all $L_p$ packets can be   decoded by $\mathcal{D}$ successfully.   
On the other hand, if the confidential packets are called for transmission, a fixed rate $R_s$ (bits per channel use) is set at $\mathcal{S}$, and $L_s=\frac{R_s}{b}$ packets from $\Pi_\mathsf{s}$  are pre-processed to be sent over the air.  
Instructions for calculating appropriate $L_p$ and $R_s$ are studied in detail during later sections.   
In order to enhance the level of security for the packet-based transmission scheme,   the raw selected source packets are further encoded into fountain-like packets in the application layer  before designing the transmitting signal at PHY   \cite{sun-Ind-relay,sun-prediction}. 
This is elaborated in the following subsection.

\subsection{Fountain-like Packet Reconstruction}	
According to the feedback of $\mathcal{D}$ to $\mathcal{S}$  during slots $1$ to $t-1$ and the fountain-like  encoding rule which is publicly shared among system nodes, 
$\mathcal{S}$ updates two sets for each of the public and confidential packets, i.e.,   a recovered set $\mathcal{R}_k$,  
and an unrecovered set $\mathcal{U}_k$, $k\in\{\mathsf{s},\mathsf{p}\}$, 
which are, respectively, composed of all the source packets that have been and have not been recovered at $\mathcal{D}$, by the $t$-th slot.  
Without loss of generality, we consider   $\mathcal{U}_{k}$ has the form $\mathcal{U}_k = \left\{\pi^k_{1}, \pi^k_{2}, \cdots, \pi^k_{|\mathcal{U}_k|} \right\}$, for $k\in\{\mathsf{s},\mathsf{p}\}$.  
In time-slot $t$, 
considering 
$L_p < |\mathcal{U}_\mathsf{p}|$ (if public packets are scheduled to be transmitted) or $|\mathcal{U}_\mathsf{s}|>L_s$ (otherwise),   
$\mathcal{S}$ calls the most recently recovered source packet $\pi_\mathsf{rec}^k \in \mathcal{R}_k$
and XORs
it with $L_k$ unrecovered source packets $\{\pi^{k}_{ind}, \cdots , \pi^{k}_{ind+L_k-1}\}$ from $\mathcal{U}_{k}$ to generate $L_k$ fountain-encoded packets, 
where $ind$ denotes the starting index of chosen packets. 
Mathematically speaking,  we have the following encoding strategies to obtain a coded packet denoted by $c_l$ at the application layer  
\begin{align}
	c_l=\begin{cases}\label{xor} 
		\pi^\mathsf{p}_\mathsf{rec}\oplus\pi^\mathsf{p}_{ind+l}, \quad \text{for } 0\leq l\leq L_p-1,  & \text{if } \Psi_0\\
		\pi^\mathsf{s}_\mathsf{rec}\oplus\pi^\mathsf{s}_{ind+l}, \quad \text{for } 0\leq l \leq L_s-1,   &\text{if } \Psi_1
	\end{cases}
\end{align}
where $\Psi_0$ and $\Psi_1$ correspond to the case of public  and confidential data transmission, respectively. Equivalently, $\Psi_0$ corresponds to the case of $\hat{{g}}_{sd}\leq\nu$, while  $\Psi_1$ corresponds to  complementary  case of $\hat{{g}}_{sd}>\nu$, with 
\begin{align}	\mathsf{Pr}(\Psi_1)\stackrel{(a)}{=}1-\mathsf{Pr}(\Psi_0)=e^{-\nu}\sum_{i=0}^{n_\mathsf{T}-1} \nu^i/i!, 
\end{align} 
which follows directly from the fact that   $\hat{g}_{sd}\treq||\hat{\boldsymbol{h}}_{sd}||^2$  obeys Gamma distribution $\mathcal{G}(n_\mathsf{T},1)$ \cite{Digital-commun}.   
We also note that if $L_p(L_s) \geq |\mathcal{U}_\mathsf{p}|(|\mathcal{U}_\mathsf{s}|)$,  all  source packets in $\mathcal{U}_\mathsf{p}(\mathcal{U}_\mathsf{s})$ are
selected to generate  fountain-coded packets similar to \eqref{xor}.

By invoking  \eqref{xor} and considering the fact that the
fountain-coded packets  will be sent over the wireless
channel (which is discussed in the next subsection), 
since the $\mathcal{S}$-to-Eves channels vary independently compared to the legitimate link,  it is unlikely for Eves to be able to decode exactly the same packet(s) as $\mathcal{D}$ does.  
Hence, missing a frame at any Eve, means the linear combination of source packet(s) included in that frame may never be seen by her again.

{\textit{Remark 1:}} \label{Remark1}  
To further exploit the content of the image file opted to be transmitted in a  delay-limited scenario, 
the source packets can be prioritized in a pre-processing step, according to their content.   
Toward this end, pixels with higher saliency values or the ones having edge contents can be considered to have higher levels of importance for the legitimate observer. 
In this regard,  a practical method for determining the level of importance for each source packet   is proposed in  \cite{sun2020}. 
Consequently, the confidential (RoI) and the public (BG) packets can be sorted,  from the highest to the lowest priority. 
Therefore, at each TS, the associated index $ind$ in \eqref{xor} can be set to $1$, indicating that we choose the $L_p(L_s)$ most important packets to be sent during each TS. 

\subsection{Signaling Design for Semi-adaptive Image Transmission}
After preparing fountain-coded packets $c_l$ at the application layer based on \eqref{xor},  
in this subsection, we elaborate on the design of transmit signal within a given frame at PHY.   
In order to safeguard data security of  confidential packets, the transmit signal is combined with AN, while public BG packets are sent directly, using a beamforming vector $\boldsymbol{w}^\mathsf{T}\treq {\hat{\boldsymbol{h}}^\dagger_{sd}}/{||\hat{\boldsymbol{h}}_{sd}||}$.  Hence, each fountain-coded packet has one of the following forms in the transmitting frame: 
\begin{align}\label{PHY-design}
\begin{cases}
\boldsymbol{x}_p=\sqrt{ P_p}\boldsymbol{w}m_p,  &\text{if } \Psi_0\\
\boldsymbol{x}_s=\sqrt{\zeta P_s}\boldsymbol{w}m_s+\sqrt{\frac{(1-\zeta)P_s}{(n_\mathsf{T}-1)}}
\mathbf{G}\boldsymbol{v}_\mathrm{AN},
&\text{if } \Psi_1  
\end{cases}
\end{align} 
where $\boldsymbol{x}_p$ and  $\boldsymbol{x}_s$ correspond to the transmitted PHY signal of  a public fountain-coded packet and a confidential  one, respectively. Moreover,  $m_k$, $\boldsymbol{v}_{AN}\in\mathbb{C}^{(n_\mathsf{T}-1)\times1}$, and $P_k$ ($k \in\{s,p\}$) denote the unit-power information-bearing signal, i.e., the modulated version of the corresponding encoded packets, the AN with i.i.d. elements $\boldsymbol{v}_i\sim \mathcal{CN}(0,1)$, and the transmit power, respectively.  
We note that   $0<\zeta<1$ shows the power allocation ratio between the information-bearing signal and the AN.  
We also note that the columns of $\mathbf{G}\in\mathbb{C}^{n_\mathsf{T}\times (n_\mathsf{T}-1)}$ and $\boldsymbol{w}$ form an orthogonal basis.  
In other words, the source node $\mathcal{S}$  adjusts its beamforming vector $\boldsymbol{w}$, based on the estimation of legitimate link $\hat{\boldsymbol{h}}_{sd}$, and injects AN to the original signal in a way that the legitimate link is not affected by the injected noise. 
Consequently, the received signal at $\mathcal{D}$ when sending  public  and  confidential data are  expressed, respectively,  as 
\begin{align}\label{y_D}
y^\mathsf{p}_D&=\sqrt{P_p}\rho r^{-\frac{\eta}{2}}_D||\hat{\boldsymbol{h}}_{sd}||m_p+\sqrt{(1-\rho^2)P_p}r^{-\frac{\eta}{2}}_D\boldsymbol{w}^\mathsf{T}\hat{\boldsymbol{e}}_{sd} m_p+n_D, 
\nonumber\\ 
y^\mathsf{s}_D&=\sqrt{\zeta P_s}\rho r^{-\frac{\eta}{2}}_D||\hat{\boldsymbol{h}}_{sd}||m_s+\sqrt{\zeta (1-\rho^2)P_s }r^{-\frac{\eta}{2}}_D\boldsymbol{w}^\mathsf{T}\hat{\boldsymbol{e}}_{sd} m_s
+ \sqrt{\frac{(1-\zeta)(1-\rho^2)P_s}{(n_\mathsf{T}-1)}}r^{-\frac{\eta}{2}}_D \hat{\boldsymbol{e}}_{sd}^\mathsf{T}
\mathbf{G}\boldsymbol{v}_\mathrm{AN}  +n_D. 
\end{align}
Based on \eqref{y_D}, the instantaneous received signal-to-interference-plus-noise ratio (SINR) at $\mathcal{D}$\footnote{{Notably, the second term in \eqref{y_D} denotes the interfered AN at the legitimate link due to channel uncertainty.}} is formulated as 
\begin{align}\label{gamma_D} 
	\hspace{-2mm}
\gamma^{\mathsf{s}}_{D}=\frac{\zeta\rho^2P_s||\hat{\boldsymbol{h}}_{sd}||^2}{(1-\rho^2)P_s+r_D^{\eta}\sigma_n}, \quad  \gamma^{\mathsf{p}}_{D}=\frac{\rho^2P_p||\hat{\boldsymbol{h}}_{sd}||^2}{(1-\rho^2)P_p+r_D^{\eta}\sigma_n}. 
\end{align}  
Invoking the expression for  $\gamma^\mathsf{p}_D$ in \eqref{gamma_D}, at the beginning of public TSs,  $\mathcal{S}$ calculates the number $L_p$ of public encoded-packets which should be prepared to construct the transmit frame as follows 
\begin{align}\label{L}
L_p=  \floor[\bigg]{\frac{BT}{{b}}\log_2(1+\gamma^\mathsf{p}_{D})}.   
\end{align} 
Obviously, \eqref{L}  ensures that all $L_p$ public packets can be successfully decoded by Bob without any error using capacity-achieving codes at PHY.   
Therefore, we can meet the delay requirements of packet-based image delivery via utilizing multi-packet transmission of public data.  
Similarly, for the transmission of confidential packets, we also opt for finding a suitable number $L_s$ of transmit  packets during each TS to satisfy the security constraints  as well as addressing the delay  limits. This is 
studied  in the subsequent section. 

In the case of transmitting confidential  data, it is also of great importance to investigate the received signal $y^\mathsf{s}_{E_i}$  and the instantaneous received SINR $\gamma^\mathsf{s}_{E_i}$ 
at passive Eves for $i \in\Phi_E$. Hence, for the proposed scenario, we have  
\begin{align}
y^\mathsf{s}_{E_i}&= \sqrt{\zeta P_s}\boldsymbol{w}^\mathsf{T} \boldsymbol{h}_{se_i} m_s+\sqrt{\frac{(1-\zeta)P_s}{(n_\mathsf{T}-1)}}\boldsymbol{h}^\mathsf{T}_{se_i}
\mathbf{G}\boldsymbol{v}_\mathrm{AN}+n_{E_i}, 
\nonumber\\
\gamma^\mathsf{s}_{E_i}&=\frac{\zeta P_s r_i^{-\eta} |\boldsymbol{w}^\mathsf{T}\boldsymbol{h}_{se_i}|^2}
{(1-\zeta)P_s||\boldsymbol{h}_{se_i}^\mathsf{T}\mathbf{G}||^2r_i^{-\eta}/(n_\mathsf{T}-1)+\sigma_n \label{gamma_Ei}
},
\end{align}
where $\boldsymbol{h}_{se_i}$ denotes the fading channel from $\mathcal{S}$ to the $i$-th Eve.

\section{Performance Analysis}\label{Analysis}	 
\subsection{Preliminaries}
In this subsection, we first provide two useful lemmas which are used during our subsequent performance analysis.  \begin{lemma}\label{lemma1}
	Considering the NCE case for the 
	randomly located Eves,  
	where the equivalent SINR of the wiretap channel is expressed by $\gamma^\mathsf{NCE}_E\treq \underset{i\in\Phi_E}{\max}  \gamma^\mathsf{s}_{E_i}$, the CDF of $\gamma^\mathsf{NCE}_E$ is expressed  as 
	\begin{align}
		F_{\gamma^\mathsf{NCE}_{E}}(\omega)=\exp
		\left( \hspace{-1mm} -\beta \lambda_E 
		\hspace{-1mm}
		\left(\zeta \frac{P_s}{\sigma_n}\right)^{\frac{2}{\eta}} 
		\hspace{-2mm}
		\omega^{-\frac{2}{\eta}} \bigg(1+\frac{1/\zeta-1}{n_\mathsf{T}-1}\omega\bigg)^
		{\hspace{-1mm}1-n_\mathsf{T}}
		\right)\hspace{-1mm},  	 
	\end{align} 
	where  $\beta=\pi\Gamma(1+\frac{2}{\eta})$,  {$\eta$, $\frac{P_s}{\sigma_n}$, and $\zeta$  denote the path-loss exponent, the transmit SNR of confidential packets, and the AN allocation ratio, respectively. } 
\end{lemma}
\begin{proof}
	See Appendix A. 
\end{proof} 

\begin{lemma}\label{lemma2}
	Considering the CE scenario for the 
	randomly distributed Eves in our system model, where the instantaneous SINR of the equivalent  wiretap channel is $\gamma^\mathsf{CE}_E\treq\sum_{i\in\Phi_E}^{}  \gamma^\mathsf{s}_{E_i}$,  
	the CCDF of $\gamma^\mathsf{CE}_E$ is given by 
	\begin{align}\label{ccdf}
		\overline{F}_{\gamma^\mathsf{CE}_{E}}(\omega)\stackrel{>}{\approx} \sum_{k=0}^{K}\binom{K}{k}(-1)^k
		\mathcal{L}_{\gamma^\mathsf{CE}_E}\left(\frac{k\varphi}{\omega}\right),  	 
	\end{align} 
	where 
	$K$ represents the number of terms used in the approximation\footnote{It was shown in \cite{laplace} that the provided approximation matches well with the corresponding exact expression for $K>5$.},  
	$\varphi=\frac{K}{\sqrt[K]{K!}}$, and 
	$	\mathcal{L}_{\gamma^\mathsf{CE}_E}(s) $ exhibits the Laplace transform of the equivalent wiretap channel,  which can be expressed as  shown in \eqref{L(s)} at the top of this page, 
	\begin{figure*}
		\begin{align}\label{L(s)}
			\mathcal{L}_{\gamma^\mathsf{CE}_E}(s) 
			&=
			\exp\left(-2\pi\lambda_E\mathcal{B} \exp({\frac{\varsigma s}{2}}) (a_1s)^{\frac{n_\mathsf{T}-1+2/\eta}{2}}
			W_{\frac{1-n_\mathsf{T}+2/\eta}{2},\frac{2-n_\mathsf{T}-2/\eta}{2}}\big(\varsigma s\big)
			\right)
		\end{align} 
		\hrulefill
	\end{figure*}
	with $\mathcal{B}=\frac{B\left(\frac{2}{\eta},1-\frac{2}{\eta}\right)}{\eta}\varrho^{\frac{1-n_\mathsf{T}+2/\eta}
		{2}}$,  $\varrho=(1-\zeta)\frac{P_s}{\sigma_n}/(n_\mathsf{T}-1)$,  $\varsigma=\zeta (n_\mathsf{T}-1)/(1-\zeta)$, and $a_1=\zeta \frac{P_s}{\sigma_n}$. 
\end{lemma}
\begin{proof}
	See Appendices B and C. 
\end{proof}

\subsection{QVP Derivation} 
To investigate the performance of our proposed packet-based delay-aware image transmission scheme, the quality-of-service violation probability (QVP) metric is examined. 
In comparison with the well-known secrecy outage probability (SOP) metric, QVP reflects a more comprehensive description regarding the system's performance in terms of the delay, the reliability, and the security level \cite{sun-Ind-relay}.  
To be specific, SOP characterizes the probability with which the file is reliably received at $\mathcal{D}$, without any information leakage at Eves. 
In contrast, QVP calculates the probability with which the file is reliably received at legitimate destination within the given delay bound, without any information leakage at Eves. 
The QVP can be  described as  
\begin{align}\label{qvp}
	\mathcal{P}_{QV}&=\mathsf{Pr}\left(
	T_D>D_{lim}
	\right)+\mathsf{Pr}\left(T_D \leq D_{lim}\right)\mathsf{Pr}\left(T_E\leq T_D\right),
\end{align}
where  $T_D$  denotes the number of TSs required by $\mathcal{D}$  to completely recover the file, and   $T_E$ denotes the number of TSs required by the set of Eves to obtain all $N_\mathrm{roi}$ confidential packets in $\Pi_\mathsf{s}$.    
Inspecting \eqref{qvp}, one can infer that the first term reflects the delay violating probability,
i.e., the probability with which the file cannot be successfully delivered from source to the legitimate  destination within the tolerable delay $D_{lim}$. 
The second term in \eqref{qvp} shows the probability of the event that Eves obtain all $N_\mathrm{roi}$ confidential packets   before the file delivery at destination is over, or  the accumulation of their corresponding  packets is accomplished simultaneously.\footnote{
	We emphasize   that  in this article, in comparison with   \cite{basic-medical} and \cite{sun2020}, we have assumed a more stringent condition on image recovery at Eves by considering the fact that recovering all RoI regions at Eves will suffice the image  delivery to be intercepted completely. {We also control the partial information leakage of RoI packets to Eves, which is addressed in the next subsection.}}
For our proposed scheme, this is interpreted as the file intercept probability (FIP), which illustrates that the image delivery is not secured although the delay limit is met. 
In other words, Eves can obtain their information regarding the entire image as soon as all confidential (RoI-related) packets have been correctly received. 
Hence, $T_E$  equals the number of TSs needed by Eve to correctly get $N_\mathrm{roi}$ RoI fountain packets.  
Summation of the two terms in \eqref{qvp} characterizes the probability with which the intended QoSec of our proposed scheme is violated. 

{\textit{Remark 2:}}\label{remark:QVP} Invoking \eqref{qvp}, one can infer that 
by tending $D_{lim}$  to infinity,  
the QVP simplifies to the special case of delay-insensitive systems.
In this case, the QVP is dominated by the second term of \eqref{qvp}, i.e., the FIP. 
This indicates that the designers  should contribute  to achieve superiority  for the legitimate link against the wiretap channels to 
decrease the FIP. 
On the other hand,  in a delay-sensitive systems with stringent delay requirements, i.e., small $D_{lim}$, the probability of $\mathcal{P}_{QV}$ is dominated by the first term of \eqref{qvp}.  
Remarkably, our proposed method
 considers both the packet accumulation enhancement at $\mathcal{D}$ via utilizing multi-packet transmission \eqref{L},  and the legitimate link improvement against Eves through injecting AN \eqref{y_D}.

In the following, we derive a closed-form expression for the QVP metric. Toward this end, we note that the discrete RV $T_D$ in \eqref{qvp} can be rewritten as 
\begin{align}\label{TD}
T_D=\overline{N}_\mathrm{bg} 
+{\frac{N_\mathrm{roi}}{L_s}}+N_D^\mathrm{out},
\end{align}  
where $\overline{N}_\mathrm{bg}$
denotes the number of TSs required for delivering all $N_\mathrm{bg}$ public packets.  
Mathematically speaking, we have 
{$\overline{N}_\mathrm{bg}\treq |\mathcal{T}_p|$,
where $\mathcal{T}_p$
is the set of TSs} used for delivering public data packets.  
By invoking \eqref{L} and the definition of $\mathcal{T}_p$, one can infer that  
$\overline{N}_\mathrm{bg}$ 
depends on the legitimate channel condition during file delivery. 
Moreover, although obtaining the  exact expression for $\overline{N}_\mathrm{bg}$  is intractable,  one can approximate  $\overline{N}_\mathrm{bg}$ as follows 
\begin{align}\label{N_tilde_bg}
\overline{N}_\mathrm{bg}&\approx 
\ceil[\bigg]{
\frac{N_\mathrm{bg}}{\sum_{k=0}^{+\infty}kp_{\mathrm{bg},k}}
}, 
\end{align}
where the denumerator denotes the expected number of fountain-coded public packets sent within a TS, and  $p_{\mathrm{bg},k}$ shows the probability with which $k$  public packets
are delivered to $\mathcal{D}$ within a slot successfully.   Thus, we have the following expression for $p_{\mathrm{bg},k}$ 
\begin{align}
p_{\mathrm{bg},k}&=
\mathsf{Pr}\Big(k\leq\frac{BT}{b}\log_2(1+\gamma^p_{D})<k+1 \Big\vert\Psi_0\Big)
\stackrel{(a)}{=} 
\left[
\left(
\gamma(n_\mathsf{T},\frac{\gamma_u}{\kappa_p})
-\gamma
(n_\mathsf{T},\frac{\gamma_l}{\kappa_p})
\right)\Big/{\Gamma(n_\mathsf{T})}  
\right] \frac{\mathbf{1}_{(\gamma_l<\gamma_u)}}
{\mathsf{Pr}(\Psi_0)},
\end{align}  
where   
$\kappa_p=\frac{\rho^2P_p}{(1-\rho^2)P_p+r_D^{\eta}\sigma_n}$, 
$\gamma_l=2^{kb/(BT)}-1$,  $\gamma_u=\min\{\kappa_p\nu,2^{(k+1)b/(BT)}-1\}$, {and $\nu$ denotes the decision  threshold in \eqref{xor}.}  We note that ($a$) follows from {\cite[Eq. (3.382.5)] {integ}}.  Moreover, it is assumed that  $\gamma_l<\gamma_u$,  
which is determined by the indicator function $\mathbf{1}_{(\gamma_l<\gamma_u)}$. 

The second term in \eqref{TD}  denotes the required TSs for sending $N_\mathrm{roi}$ confidential packets.  
Finally, $N_{D}^\mathrm{out}$  shows the number of outage events for the legitimate link of \src-to-$\mathcal{D}$.   
{By defining  $\Upsilon(k,m;p)=\binom{k+m-1}{k}p^k(1-p)^m$}, the following propositions facilitate the derivation of  QVP. 
\begin{proposition}
	The PMF of $T_D$ can be expressed as 
	\begin{align}\label{pmf_TD}
		f_{T_D}(k)&=\mathsf{Pr}\left(
		T_D=k\right)
 = {\Upsilon\Big(k-\tilde{N},\frac{N_\mathrm{roi}}{L_s};\Omega\Big)}, 		
		k\geq \tilde{N},   
	\end{align}
	where $\tilde{N}\treq \overline{N}_\mathrm{bg}+\frac{N_\mathrm{roi}}{L_{s}}$, 
	$\Omega = \frac{\mathbf{1}_{(\theta\geq\kappa_s\nu)}}	{\mathsf{Pr}(\Psi_1)}
	\Big(\gamma(n_\mathsf{T},\frac{\theta}{\kappa_s})-\gamma(n_\mathsf{T},\nu)\Big)\Big/{\Gamma(n_\mathsf{T})}$, $\kappa_s=\frac{\rho^2\zeta P_s}{(1-\rho^2)P_s+r_D^{\eta}\sigma_n}$, and $\theta = 2^{\frac{L_sb}{BT}}-1$. 
\end{proposition}
\begin{proof}
	See Appendix D. 
\end{proof} 
 
As a result, we can derive the expression for the first term of QVP in \eqref{qvp} as 
\begin{align}\label{TD>Dlim}
\hspace{-0.0mm}
\mathsf{Pr}\left(
T_D>D_{lim}	\right)
&=1
\hspace{-0.5mm}-\hspace{-0.5mm}
\sum_{k=\tilde{N}}^{D_{lim}}
\hspace{-1mm}
{\Upsilon\Big(k-\tilde{N},\frac{N_\mathrm{roi}}{L_s};\Omega\Big)}
\hspace{-1mm}. 
\end{align}

\begin{proposition}
We have the following distribution for $T_E$ 
\begin{align}\label{TE<k}
	\mathsf{Pr}(T_E\leq k)&=
	\sum_{l=\frac{N_\mathrm{roi}}{L_{s}}}^{k}
	{\Upsilon\Big(l-\frac{N_\mathrm{roi}}{L_s},\frac{N_\mathrm{roi}}{L_s};\Lambda_c\Big)},
\end{align}
where $c\in \{\mathsf{NCE}, \mathsf{CE}\}$ reflects which type (CE or NCE scenario) of eavesdropping strategy is run by Eves through the network, $\Lambda_\mathsf{NCE} = \mathsf{Pr}(\Psi_1)\exp\Big(
\hspace{-2mm}-\beta \lambda_E (\zeta P_s/\theta)^{\frac{2}{\eta}} (1+\frac{1/\zeta-1}{n_\mathsf{T}-1}\theta)^{1-n_\mathsf{T}}
\Big) 
+\mathsf{Pr}(\Psi_0)$, and $\Lambda_\mathsf{CE} = \mathsf{Pr}(\Psi_1)
\left[1-\sum_{k=0}^{K}\binom{K}{k}(-1)^k
\mathcal{L}_{\gamma^\mathsf{CE}_E}\left(\frac{k\varphi}{\theta}\right)\right]  + \mathsf{Pr}(\Psi_0)$. 
\end{proposition} 
\begin{proof}
See Appendix E. 
\end{proof}

\vspace{-2mm}

Based on the provided discussions in this subsection, by invoking \eqref{qvp}, one can rewrite the QVP metric as 
\begin{align}\label{qvp-2}
\mathcal{P}_{QV}=\mathsf{Pr}\left(
T_D>D_{lim}
\right)+\sum_{k=\tilde{N}}^{D_{lim}}
\mathsf{Pr}\left(T_D=k\right)\mathsf{Pr}\left(T_E\leq k\right).  
\end{align}
Therefore, substituting  \eqref{pmf_TD}, \eqref{TD>Dlim}, and \eqref{TE<k} into \eqref{qvp-2}  completes the derivation of QVP for our proposed scheme.

\subsection{Data Rate Adjustment to Assure Security for Confidential Packets}
In our  proposed scheme, to meet the delay constraints in a practical transmission scenario, the number of public packets that are called to be sent during each slot is carefully set owing to the multi-packet adaptive transmission of public packets (See Eq. \eqref{L}).    
The focus is now shifted toward the procedure of delivering  confidential packets. 
{With this regard,  we aim to control the partial information leakage of confidential packets to Eves. 
Mathematically speaking,} we ensure the intercept probability (IP) of sending confidential packets to be less than a desired threshold $\epsilon_\text{IP}$.  
The IP metric  describes the probability with which Eves intercept the transmitted confidential frame, given the transmission rate of source.   
Accordingly, by adjusting $\epsilon_\text{IP}$, one can achieve a desired security level for the proposed scheme. 
Based on the aforementioned discussions, one can satisfy the following inequality 
\begin{align}\label{SOP}
	\mathcal{IP} & \treq  
\mathsf{Pr}(\Psi_1)
\mathsf{Pr}\bigg\{\frac{BT}{b}\log_2(1+\gamma^c_E)>{L}_{s}	\Big\vert\Psi_1\bigg\}
\leq \epsilon_\text{IP}. 
\end{align}
where $c\in\{\mathsf{NCE},\mathsf{CE}\}$. 
Then, \eqref{SOP} can be simplified to obtain an acceptable range for the number, ${L}_{s}$, of confidential packets  that should be chosen  to form confidential data frame. 
 Hence, the following proposition is provided.  
\begin{proposition}\label{prop1}
The acceptable range for the number ${L}_{s}$ of  packets  comprising confidential data, based on limiting the transmission scheme  to have an IP less than or equal to $\epsilon_\text{IP}$ is given by  
\begin{align}\label{L_range1}
	{L}_{s}\geq \frac{BT}{b}\log_2\left(1+F_{\gamma^c_E}^{-1}
	\left(1-\frac{\epsilon_\text{IP}}{\mathsf{Pr}(\Psi_1)}\right)
	\right), 
\end{align}
where $F_{\gamma^c_{E}}$ 
is derived in Lemma \ref{lemma1} and \ref{lemma2}.  
\end{proposition}
The proof of Proposition \ref{prop1} can be obtained  through a straightforward manipulation on the definition of IP in \eqref{SOP} and using the fact that the CDF $F_X(x)$ is a monotonically increasing function on $x$. 


\section{Optimized Learning-Based Image Transmission} 
In this section, we aim to optimize the network parameters  to enhance the overall system performance.   
By invoking \eqref{pmf_TD}--\eqref{qvp-2}, one can infer that the QVP is affected by different channel conditions, as well as the image size and the location of the legitimate node being serviced.   
More specifically,  every time a new image is scheduled to be transmitted to the legitimate receiver $\mathcal{D}$, or when a new device is authenticated to 
receive packets from $\mathcal{S}$,  the source  
should adjust its transmission parameters, including the power allocation ratio $\zeta$ between information and AN, the transmission power for public and confidential packets, denoted by $P_p$ and $P_s$,  respectively, and the decision threshold $\nu$ on whether to send public or confidential packets. 
The  transmission rate $R_s=L_sb$ of confidential packets 
can be  optimized as well.   
Moreover, the transmitter should be robust against different states of channel imperfection, which is modeled by parameter $\rho$ in this paper.      

Mathematically speaking, for each complete round of image delivery from $\mathcal{S}$ to a legitimate receiver,   the optimal transmission parameters  $(\zeta^\ast,P^\ast_p,P^\ast_s,\nu^\ast)$
to minimize the QVP metric
can be found from the following optimization problem   
\begin{equation}
\begin{aligned}\label{opt}
\underset{\zeta,P_p,P_s,\nu,L_{s}}{\textrm{minimize}} \quad & \mathcal{P}_{QV}
\\
\textrm{s.t.} \quad  
& \mathcal{IP} \leq \epsilon_\text{IP},  \hspace{2mm}
0 <\zeta \leq 1, \hspace{2mm} \nu>0, \hspace{2mm} 
\gamma_\mathrm{min}<\frac{P_s}{\sigma_n},
\frac{P_p}{\sigma_n}\leq \gamma_\mathrm{max},  \hspace{2mm}
L_{s} \vert  N_\mathrm{roi},     
\end{aligned}
\end{equation}
where $\gamma_\mathrm{min}$ and   $\gamma_\mathrm{max}$  represent the minimum and the maximum available transmit SNR at the source node $\mathcal{S}$.   
In addition,  the last constraint indicates that $L_s$ is a divisor of $N_\mathrm{roi}$.   

The proposed optimization problem  is a mixed-integer non-linear programming (MINLP) due to its non-linear  (and highly non-convex) objective and constraints.   
Therefore, obtaining an analytic expression for the optimal transmission parameters $(\zeta^\ast,P^\ast_p,P^\ast_s,\nu^\ast,L_s^\ast)$ is intractable. 
Traditionally, \eqref{opt} could be solved   via numerical algorithms  leading to a large computational complexity.  
In contrary, we propose a learning-based scheme to solve \eqref{opt} in an  efficient way.  

We now elaborate on our learning-based scheme to effectively attain optimal transmission parameters.  
In our proposed scheme, $\mathcal{S}$ utilizes a deep neural network (DNN) to learn the non-trivial mapping from the 
system's configurations, i.e.,  
input image size, including the confidential and public packets, the location of the legitimate receiver, and  the  channel statistics to the optimal transmission parameters that minimizes the QVP. 
This  mapping can be formulated as
\begin{align}\label{func_approximator}
\left(\zeta^\ast,P^\ast_p,P^\ast_s,\nu^\ast,L_s^\ast\right)
=\mathcal{F}
\left(
 N_\mathrm{roi}, N_\mathrm{bg}, 
 r_D, \rho
\right). 
\end{align} 
In the following, we propose our DNN's architecture for solving \eqref{func_approximator}, which is composed of fully-connected  layers in a feedforward network. 
It is worth mentioning that the DNNs, 
specifically the feedforward neural networks,  
have shown to be capable of performing very complex tasks and obtain an input-output map that approximates any measurable function  \cite{Goodfellow}. 
Therefore, utilizing DNNs is an efficient strategy to solve the problem in \eqref{opt}.  



\subsection{Network Architecture}  
We first provide a brief overview of  feedforward DNNs.   
  A general architecture for a feedforward neural network with fully-connected layers
   is composed of an input layer, an output layer, and $K$ hidden layers.    The output of each layer is the input of its sequential layer.  
  The $k$-th layer, $k=1, \ldots, K+1$, has $u_{k}$ neurons, and the $u_k \times 1$ output vector of the $k$-th layer ($2\leq k\leq K+2$)
  can be written  as
\begin{align}\label{TransferFunction}
	\boldsymbol{x}_{k}=f_{k}\left(\mathbf{W}_k	\boldsymbol{x}_{k-1}+\boldsymbol{b}_k\right), 
\end{align}
where $f_k$ is the activation function  of layer $k$, $\mathbf{W}_k$ and $\boldsymbol{b}_k$ respectively denote the weight matrix and the bias vector of the $k$-th layer. 
Notably, $\boldsymbol{x}_0$ is the input vector to the DNN. 
According to \eqref{TransferFunction}, each neuron is responsible for computing  quite lightweight operations. 
 The combination of  multiple neurons through stacked fully-connected  layers, the proposed DNN can   obtain an overall input-output mapping to emulate desired functions {\cite [Theorem 1] {universal}}.  
We can now elaborate on our proposed DNN for the content-aware wireless image transmission scheme.     
The problem is to fine-tune the weights and biases of \eqref{TransferFunction} to efficiently estimate \eqref{func_approximator}.


Considering an E2E realization of image delivery from the source node, $\mathcal{S}$, to a legitimate destination $\mathcal{D}$, we utilize a feedforward DNN, in which the following    information are fed into the network as input:‌ 
The number of confidential and public source packets; the level of wireless channel imperfection estimated at $\mathcal{S}$, which was modeled via the  parameter $\rho$ in \eqref{channel-imp}; and, the location of the  legitimate destination $\mathcal{D}$.      
We opt for producing  vector $\boldsymbol{x}_{K+2} = [\zeta^\ast,P^\ast_p,P^\ast_s,\nu^\ast,L_s^\ast]^T$ corresponding to the optimal transmission parameters at the output layer.    
As illustrated in Fig. \ref{fig:DNN-proposed}, 
between the input and output layers, $K=4$ dense layers are deployed each of which contain a fully-connected layer followed by a batch normalization (BN) block. 
Notably, the BN is performed with the aim of enhancing the generalization properties of our DNN, and making the learning process 
more stable.   {The BN block performs $\hat{{x}}_{k,i} = \frac{{x}_{k,i} - \mu_{\mathcal{B}}}{\sqrt{\sigma^2_{B}+\epsilon}}, 
\quad  
{y}_{k,i}=\lambda_i \hat{{x}}_{k,i} + \delta_i,$
on its input $\boldsymbol{x}_k$, 
$k\in\{2\cdots,K+1\}$
for all $\boldsymbol{x}_k$'s in a mini-batch,}  
where ${x}_{k,i}, i \in \{1,\cdots,u_{k}\}$ is the $i$-th input  of BN block at the $k$-th layer, $y_{k,i}$ is the corresponding $i$-th output of the BN block,  $\lambda_i$ and $\delta_i$ are  learnable parameters along with the other network parameters.   
Moreover, $\mu_\mathcal{B}$ and $\sigma^2_\mathcal{B}$ are the mean and variance of training data over a mini-batch of size $m$, while  $\epsilon$ is for the sake of stability.       
 Activation functions of hidden neurons in our proposed DNN are chosen to be rectified
 linear unit (ReLU), which is defined as 
 $f^\mathsf{ReLU}_k(x)\treq\max(0,x)$, $2\leq k \leq K+1$.  
 The output layer is obtained from a fully-connected layer followed by a linear activation function.      

 
\subsection{Training of the Proposed DNN: Train Offline, Use Online} 
In order to find an appropriate estimation for \eqref{func_approximator},  the weights and biases in \eqref{TransferFunction} 
need to be adjusted. 
Accordingly, the weights $\mathbf{W}=[\mathbf{W}_2,\cdots,\mathbf{W}_{K+2}]$ and the biases $\mathbf{B}=[\boldsymbol{b}_2,\cdots,\boldsymbol{b}_{K+2}]$ as well as the scale and shift parameters of BN blocks in hidden layers, denoted respectively by $\mathbf{\Lambda} = [\boldsymbol{\lambda}_2,\cdots,\boldsymbol{\lambda}_{K+1}]$ and $\mathbf{\Delta} = [\boldsymbol{\delta}_2,\cdots,\boldsymbol{\delta}_{K+1}]$,  
are configured. 
This adjustment is carried out in a supervised manner via training our DNN with a training set 
$\mathcal{T} = \left\{ ( \boldsymbol{r}_n,  \boldsymbol{p}_n^{\ast}) \right\}$, $n = 1, \cdots, \mathsf{N}_\mathcal{T}$, with $\mathsf{N}_\mathcal{T}=|\mathcal{T}|$ training tuples, where $\boldsymbol{p}_n^{\ast} = 
[\zeta^\ast,P^\ast_p,P^\ast_s,\nu^\ast,L_s^\ast]^T$ is the vector of 
desired
transmission parameters corresponding to a realization  of system with configuration  parameters $\boldsymbol{r}_n = [N_\mathrm{roi}, N_\mathrm{bg}, 
r_D, \rho]^T$. 
 Note that  the training set can be obtained via prevalent numerical methods for solving an optimization problem through an offline phase {\cite{learning-based,Besser_Opt}}. 
Using the examples provided in $\mathcal{T}$, the DNN gradually  learns to predict the trasmission parameters for new realizations of $\boldsymbol{r}_n$  as well. 
 Mathematically speaking, the training process opt for adjusting the weights and biases of our DNN with the goal of  minimizing the loss between actual and desired output vector, which is formulated as follows
\begin{align}\label{TrainingMin}
	\hspace{0mm}
\underset{{\mathbf{W,B,\Lambda,\Delta}}}{\textrm{minimize}}	
\hspace{2mm}
\frac{1}{\mathsf{N}_\mathcal{T}}
\hspace{-1mm}
\sum_{n=1}^{\mathsf{N}_\mathcal{T}}\hspace{-1mm}
{\ell}\left(\boldsymbol{x}^{(n)}_{K+2}(\mathbf{W,B,\Lambda,\Delta}), {\boldsymbol {p}}_{n}^{*}\right), 
\end{align}
where $\boldsymbol{x}^{(n)}_{K+2}(\mathbf{W,B,\Lambda,\Delta})$  denotes the output of DNN corresponding to the $n$-th training input,  $\boldsymbol {p}_{n}^{*}$ is the desired output, and   ${\ell}(\cdot,\cdot)$ is any desired  error measure between these two.  
 We employ mean-squared-error (MSE) 
 $||\boldsymbol{x}^{(n)}_{K+2}(\mathbf{W,B,\Lambda,\Delta})-{\boldsymbol {p}}_{n}^{*}||^2$
 as a widely-used error measure  in this paper.  
  The minimization of \eqref{TrainingMin} can be handled  by 
  off-the-shelf  gradient descent-based methods specifically developed for training DNNs \cite{Goodfellow}, which is not reviewed here.\footnote{
  	We have chosen the widely-adopted  adaptive moment estimation (Adam) optimizer algorithm for minimizing \eqref{TrainingMin} \cite{adam}.  The convergence of our training process is validated in the simulation results.}

After that the training phase is completed, i,e, the minimization problem of \eqref{TrainingMin} converges to a relatively low MSE,  
our DNN achieves an acceptable approximation for the  mapping in \eqref{func_approximator}. 
 Afterward,  when a new configuration, ${\boldsymbol{r}^\mathsf{new}}$, is defined for the system, i.e., a new image is scheduled to be sent or a new legitimate node is signed in to the network,  
 the corresponding transmission parameters can be obtained without the need to solve \eqref{opt} again. 
Instead,  it suffices to compute the output of the trained DNN, using a forward propagation  with the new input ${\boldsymbol{r}^\mathsf{new}}$ in a real-time manner.  

{\textit{Remark 3:}}
Compared with the conventional approaches which mandate the system to perform iterative methods from scratch---which actually imposes much delay on practical systems---every time one or more configurations are altered (e.g.,  the source should send an image),    
our proposed procedure brings about  a significant  decrease in computational complexities during the online phase of using the trained  DNN. 
This is actually in line with the requirements of  our practical delay-aware model.    
More details are addressed in the following subsection. 
 
\subsection{Computational Complexity}
Based on the provided discussions in Section V-B, we first stress that once the weights and biases in \eqref{TrainingMin} are determined, the input-output relationship of the DNN can be calculated as the composition of  affine combinations and activation functions of the neurons as proposed in \eqref{TransferFunction}. 
{In other words, the main advantage
 of our proposed learning-based method is that one can perform most of the computations required for solving \eqref{opt} offline, and   
 only a few operations  need to be  calculated  when the system configurations vary.} 

\subsubsection{Online phase}
When the trained DNN is exploited for predicting new outputs, it is required to compute the output  of each neuron in the DNN, moving forward  from the input layer to the output layer. 
 Therefore,  by invoking  \eqref{TransferFunction} and 
 the BN adjustment formula, 
 one can argue that  $\sum_{k=2}^{K+2}u_{k-1}u_{k}+\sum_{k=2}^{K+1}u_k$ real multiplications, together with calculating  $\sum_{k=2}^{K+2}u_{k}$ scalar activation functions $f_{k}$ are incurred during online computation.\footnote{The incurred complexity of additions is neglected compared with the complexity of 
 	 multiplications.}  
 Notably, the activation functions are elementary functions which do not impose any significant computational complexities to the network.  
 Consequently, finding the output of our trained DNN for a given realization vector comprises  negligible complexity as it simply requires the calculation of  forward propagation through the trained DNN.

\subsubsection{Offline phase}
The offline phase, as proposed in Section V-B, includes training set generation and its utilization  to train the network.
 The training procedure can be carried out efficiently, using off-the-shelf stochastic gradient descent algorithms with  fast convergence \cite{Goodfellow}.
  Generating the training set $\mathcal{T}$ requires  the proposed problem in \eqref{opt} to be solved $|\mathcal{T}|$ times  for different realizations of the system parameters.\footnote{{We use the genetic algorithm \cite{ga} for numerically  solving \eqref{opt} and obtaining an experimental training set $\mathcal{T}$.  This is elaborated in the next section.}}  
  Although it seems to be in contradiction with  the purpose of using DNNS, but this is not the case for the following reasons:
\begin{itemize}
		\item[(i)] The training phase---including the data set generation and DNN adjustment---is performed  \emph{offline}. 
	In this manner, a much higher computational time and complexity can be afforded with significantly less constraints than a real-time computation \cite{Besser_Opt}.  
	\item[(ii)] The update  procedure of training set is done sporadically. In other words, generating training samples  can be done with a \emph{much greater time-scale} than the real-time configuration settings.  
	
\end{itemize}


\section{Numerical Results and Discussions}
In this section, we present several numerical examples to verify our derived  closed-form expression for the QVP metric.   Moreover, some relevant benchmarks are compared with our proposed scheme to show the efficiency of our hybrid transmission scheme. 
The impact of multiple randomly-located Eves on the performance of our proposed scheme together with some insights on design parameters are also provided. 
The convergence of our ML-based optimization problem is validated by investigating the training performance of our proposed DNN.  
Numerical results show the effectiveness of utilizing DNN for employing optimal transmission parameters.    
The codes are implemented in MATLAB, and 
 were run on  a 64-bit 1.80 GHz  Intel(R) Core(TM) i7-8550U CPU. 
Moreover, the offline phase of data set generation and DNN training were run on Intel(R) Xeon(R) Silver 4114  CPU running at 2.20 GHz.  
In order to accelerate the offline computations, the workflow was scaled up by leveraging a parallel pool of multiple distributed workers 
\cite{parpool}.

For the following experiments,  similar to \cite{sun2020} and \cite{sun-prediction},  we consider the nodes to be located on a normalized two-dimensional region, where $\mathcal{S}$ is placed at the origin without loss of generality.   
{The passive Eves are assumed to be distributed according to homogeneous  PPP $\Phi_E$ with density  $\lambda_E=0.2$} (except for Fig. \ref{fig:Lambda} in which we sweep the value of $\lambda_E$ to see its effect on system performance).  
We set the path loss exponent to $\eta=4$,  and $\frac{BT}{b}=\frac{50}{8}$ during the simulations  \cite{sun2020}.   {For the following figures, the source node is assumed to be equipped with $n_\mathsf{T}=8$ transmit antennas}, except for Fig. \ref{fig:EIP} in which we examine the  effect  of $n_\mathsf{T}$.    
For Figs. \ref{fig:bench}--\ref{fig:Lambda}, we consider the legitimate receiver $\mathcal{D}$ to be located at  $[2,-2]$. 
We also suppose $N_\mathrm{roi}=300$ and $N_\mathrm{bg}=200$ confidential and public packets, respectively, are  scheduled to be delivered with the transmit SNR of $30$ dB, while  the allocation ratio between the information signal and the AN is considered $\zeta=0.5$.  
The fixed rate of sending  confidential packets is also set to $L_s=20$ for the first three simulation figures. 
 We stress that the optimal transmission parameters, including the optimal transmit SNRs together with  the power allocation ratio, and the optimal  confidential transmission rate are evaluated and examined  in the subsequent Figs.   \ref{fig:TrainRmse}--\ref{fig:qvp_rho}.   

\begin{figure}
		\centering
	\begin{minipage}{0.45\textwidth}
		\centering
	\includegraphics
	[width=3.5in,height=2.5in,
	trim={0.1in 0.0in 0 0.2in},clip]{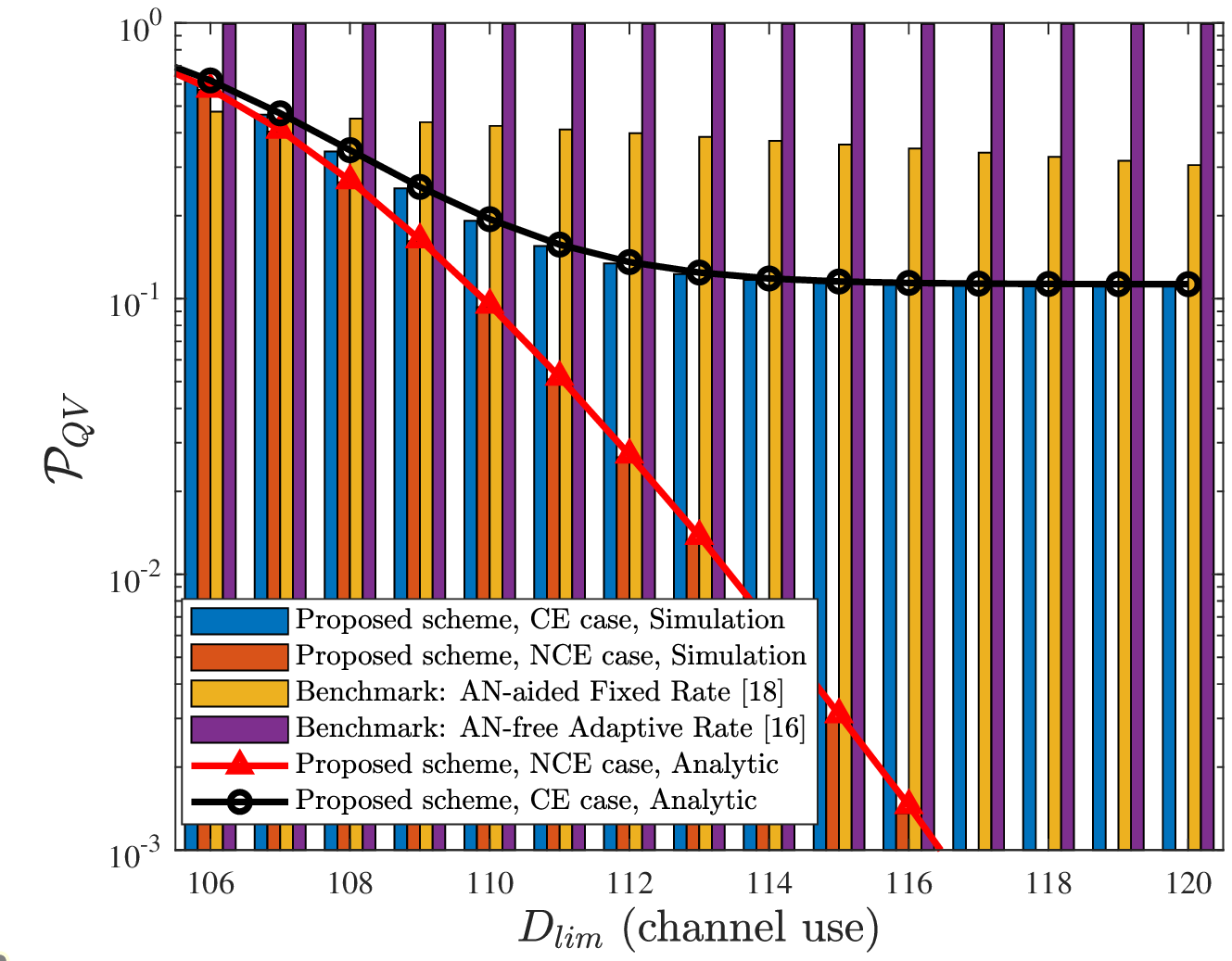}
	\caption{\small QVP versus $D_{lim}$ for our proposed scheme.}
	\label{fig:bench}
	\vspace{-7mm}
\end{minipage}\hfill
\begin{minipage}{0.45\textwidth}
\centering
		\includegraphics
	[width=3.4in,height=2.4in,
	trim={0.1in 0.0in 0 0.2in},clip]{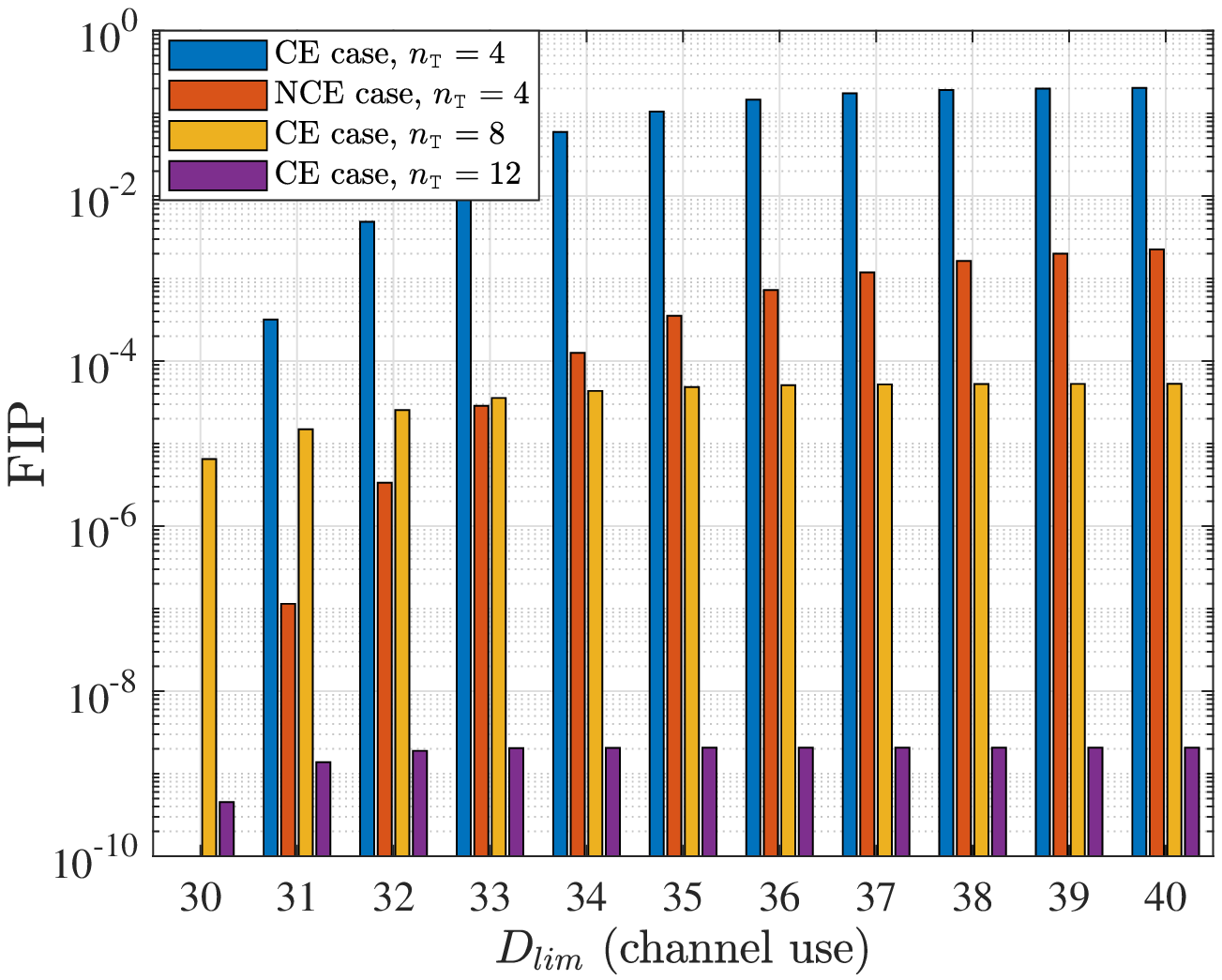}
	\caption{\small The file intercept probability vs. $D_{lim}$ for different values of $n_\mathsf{T}$.}
	\label{fig:EIP}\vspace{-7mm}
\end{minipage}
\end{figure}

Fig. \ref{fig:bench} illustrates the QVP versus the tolerable delay limit $D_{lim}$ for our proposed scheme  and two benchmarks.   
 For this figure,  the decision level $\nu$ is set to $6$, and  
 the channel correlation coefficient is set to $\rho=0.95$ \cite{laplace}. 
  A security level of $\epsilon_\text{IP}=0.1$ is also imposed on the transmission of confidential packets based on \eqref{SOP}. 
According to the figure, one can easily verify  that for all values of $D_{lim}$, the experimental results  match well with the analytical closed-form expressions obtained for the QVP in \eqref{qvp-2}. 
 The numerical results of this figure  are obtained using Monte Carlo simulation over $10^5$ realizations.  
Based on Fig. \ref{fig:bench}, it is clear that  $\mathcal{P}_{QV}$ starts to decrease with the increase in  $D_{lim}$.  
 Generally speaking, by relaxing the delay limit it becomes easier for the legitimate destination node to accumulate  $N=N_\mathrm{roi}+N_\mathrm{bg}$ coded packets within the delay bound.   
More specifically, for smaller values  of $D_{lim}$, the QVP is   reduced with the increase in $D_{lim}$.
 This is because, in this case, the QVP is diretly  dominated by the delay violating probability as discussed in Remark 2.   
However, in further increments of $D_{lim}$,  $\mathcal{P}_{QV}$
is also affected by the presence of  {totally passive} Eves trying to improve their FIP.  Therefore, the QVP might become saturated. 
This fact can be explicitly  seen from Fig. \ref{fig:bench}, where for the CE strategy, in which the equivalent wiretap channel is the accumulation of all individual eavesdropping channels, the saturation occurs at moderate values of $D_{lim}$. 
However, for the NCE case, in which the eavesdroppers do nor collude with each other, the monotonic decrease in QVP can be seen in the range of $106<D_{lim}<120$.  
This can be deduced from equations \eqref{TE<k} and \eqref{qvp-2}, together with Lemmas \ref{lemma1} and \ref{lemma2}. 
Fig. \ref{fig:bench} also demonstrates the outperformance of our hybrid approach for image delivery in comparison with two benchmarks: 1) The adaptive rate transmission with no AN injection (ARNAN) \cite{sun2020}; and, 2) The fixed-rate  (multi-packet) AN-aided  (FRANA)  transmission \cite{sun-Ind-relay}\footnote{We note that the authors in \cite{sun-Ind-relay} actually considered only one packet per TS in their AN injection scheme. However, for the sake of fair comparison and without loss of generality, we consider a fixed number of packets to be transmitted in benchmark 2.}.  
We can observe from the figure that  the ARNAN case cannot achieve small values of QVP within the given range of delay limits in this figure, which means the  utilization of adaptive multi-packet transmission is not solely sufficient to conquer wiretappers.  
This highlights the importance of AN injection in a wireless transmission environment comprising of multiple passive Eves.    
In contrast, the QVP of FRANA benchmark starts to decrease with a relatively small slope as the delay limit increases. In other words, with the increase in $D_{lim}$, where the existence of multiple Eves starts to affect the system's performance, the injected AN can help achieve smaller values of $\mathcal{P}_{QV}$.  
However, due to the fact that an adaptive transmission is not considered for this benchmark, the packet accumulation  at $\mathcal{D}$ is not properly addressed.  
On the contrary,  our proposed scheme addresses both the packet accumulation at $\mathcal{D}$, via employing \eqref{L}, and the AN injection for the confidential packets leading to a much lower $\mathcal{P}_{QV}$.  


To investigate the screcy of our proposed wireless image delivery, Fig. \ref{fig:EIP} demonstrates the FIP metric for different values of $n_\mathsf{T}$. For this figure, we have $\epsilon_\text{IP}=0.1$ and $\rho=0.95$.  Notably, we can see from the figure that  our proposed scheme can achieve arbitrary small values of FIP, indicating the secrecy of our  approach. Moreover, by increasing the number of transmit antennas deployed at $\mathcal{S}$, lower FIPs are achievable, owing to the establishment of  beamforming vectors with pencil-sharp beams which further degrades the wiretap channels according to \eqref{gamma_Ei}.    This figure also validates the discussions provided in Lemma 2 regarding the saturation of FIP by tending $D_{lim}$ to infinity.  
Therefore, one  can control the FIP to a certain limit  by carefully designing the transmission parameters,  
even though a considerable time is allowed for Eves to intercept the file.  

Fig. \ref{fig:Lambda} shows  $\mathcal{P}_{QV}$ versus the density of distributed Eves considering the security level $\epsilon_\text{IP} = 0.01$ for confidential packets.  
Different values of channel imperfection are considered in this figure, where better performance can be achieved by having more accurate estimates about the legitimate link.    
We can see from the figure that the population density of Eves can highly affect the performance of the system when  Eves are running in colluding mode. Accordingly, increasing $\lambda_E$ can result in tending the QVP to $1$ in CE case.  
This is because, for the CE scenario, the equivalent wiretap channel cumulatively  depends on all individual eavesdropping channels, i.e., $\gamma^\mathsf{CE}_E\treq\sum_{i\in\Phi_E}^{}  \gamma^\mathsf{s}_{E_i}$.  Thus, increasing $\lambda_E$, which likely increases the expected number of passive Eves, directly leads to  the increase in FIP; thus the QVP.  
On the other hand, increasing $\lambda_E$ does not directly  affect the QVP in NCE case because the equivalent wiretap channel depends only on the strongest eavesdropping link. The slight increase of QVP ‌in the NCE case is due to the fact that increasing $\lambda_E$ will increase the probability of having an Eve with a high-quality link.   {Thus, the saturated  QVP of the NCE case can be viewed as a reference range of QVP achievable for the proposed scheme. }




	\begin{minipage}{\textwidth}
		\begin{minipage}[]{0.49\textwidth}
			\centering
			\includegraphics
		[width=3.0in,height=2.2in,
		trim={0.1in 0.1in 0 0.1in},clip]{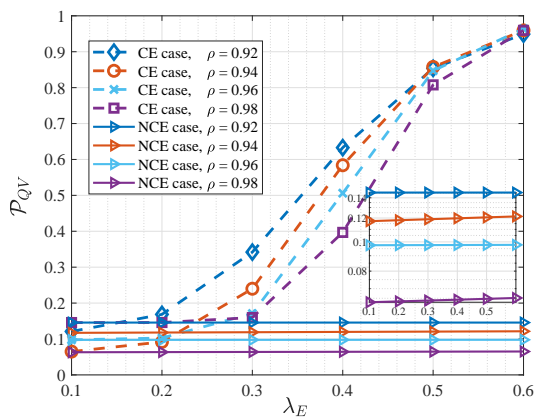}
			\captionof{figure}{\small The QVP  vs. density of distributed Eves $\lambda_E$ for different values of $\rho$}
			\label{fig:Lambda}\vspace{-3mm} 
		\end{minipage}
		\hfill
		\begin{minipage}[]{0.49\textwidth}
			 	\small
			\begin{tabular} 
				{|p{2.2in}|p{16mm}|}
				\hline \textbf{Learning Parameters} & \textbf{Values}\\
				\hline
				\hline
				Mini-batch size (${m}$) & 50\\
				Maximum number of training epochs & 500\\
				Initial learning rate & 0.001\\
				Learning rate drop factor & 0.9\\
				Number of training samples ($|\mathcal{T}|$) & 3500\\
				Number of validation samples & 750\\  
				Number of neurons in the hidden layers  & (32,16,16,8)\\ 
				Optimizer & Adam \cite{adam}\\
				\hline
			\end{tabular} \vspace{-3mm}	
			\captionof{table}{\small Parameters for Training the Proposed DNN}\label{Tab1}
		\end{minipage}
	\end{minipage}

\subsection*{{Performance of the Proposed DNN:}}
In what follows, we investigate the training process of our proposed DNN.  Moreover, the optimized performance of our learning-based scheme is also investigated.  
For the subsequent numerical experiments,  we set $D_{lim}=30$,  $\epsilon_\text{IP}=0.2$, $\gamma_\mathrm{min}=10$ dB, and $\gamma_\mathrm{max}=30$ dB. 
The learning parameters used during the training process are summarized in Table \ref{Tab1}.     
A dataset of size $5000$ is generated using the genetic algorithm \cite{ga} for solving \eqref{opt} 
	for the general parameters of  $n_\mathsf{T}=8$, $\lambda_E=0.2$,  $D_{lim}=30$. 
	The learning process is facilitated via  normalizing the transmit SNRs
	More specifically, we have  considered the  change  of variable $\tilde{P}^\ast_k\sigma_n \gamma_\mathrm{max} = P^\ast_k$, for $k = \{s,p\}$.  
	Similar reformulation is done for $L^\ast_s$ by performing $N_\mathrm{roi}\tilde{L}^\ast_s = {L}^\ast_s$ to lie in the interval $[0,1]$.

Fig. \ref{fig:TrainRmse} exhibits the training performance of our proposed DNN by investigating the average training and validation losses over epochs.  
Notably, we used a validation set  during training to verify the generalization performance of our DNN.\footnote
{Due to the fact that during the training process some  information   about the validation set leaks to the DNN, we also had another set, i.e., the test set, for the final testing of the DNN.  
The validation and test sets are 
 generated independently  
in the same manner as for the training set.}  
One can easily observe that the minimization problem \eqref{TrainingMin} corresponding to the training process quickly convergences to a small value.  
 Interestingly, both losses do not increase over  epochs which means,  the employed training process is not confronted with overfitting/underfitting phenomena.  


\begin{figure}
	\begin{minipage}{0.45\textwidth}
		\centering
	\includegraphics
	[width=3.2in,height=2.3in,
	trim={0.1in 0.15in 0 0.15in},clip]{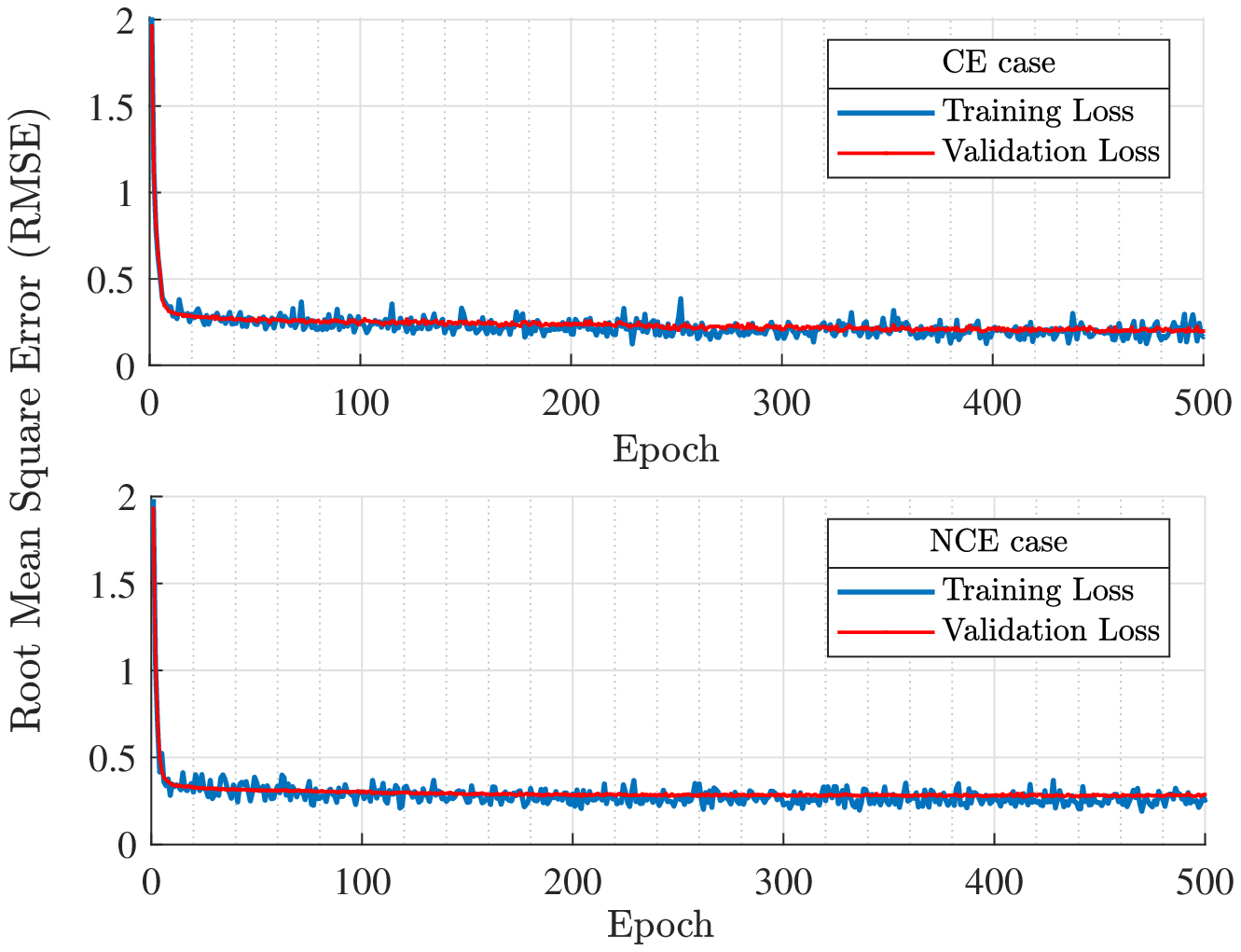}
	\caption{\small Training and validation loss of the proposed DNN.}
	\label{fig:TrainRmse}\vspace{-7mm}
	\end{minipage}\hfill
\begin{minipage}{0.45\textwidth}
\centering
	\includegraphics
	[width=3.2in,height=2.4in,
	trim={0.1in 0.1in 0 0.0in},clip]
	{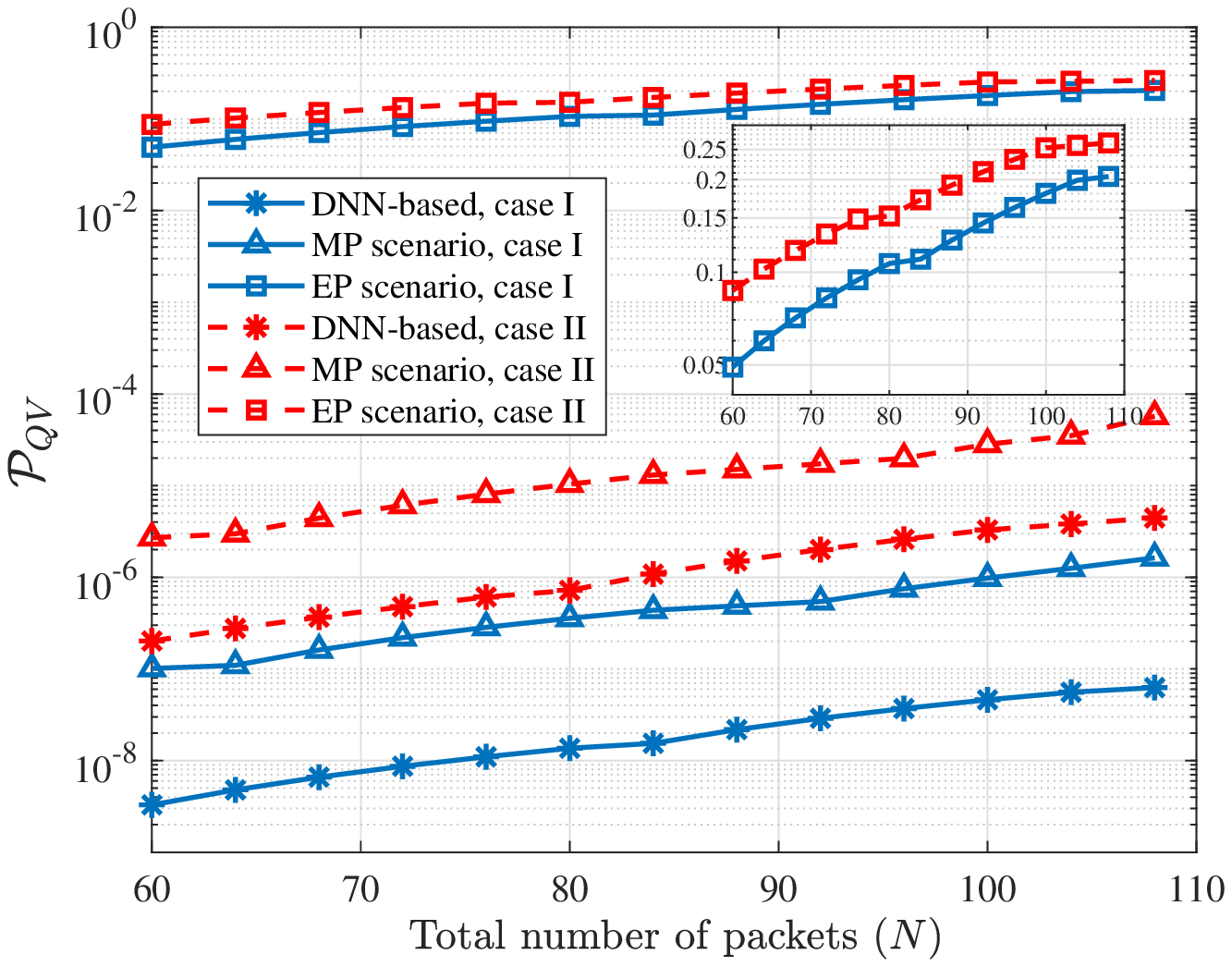}
	\caption{\small QVP vs. the total number of packets for the NCE scenario with $\rho = 0.9$. }
	\label{fig:qvp_N_NCE}\vspace{-7mm}
	\end{minipage}
\end{figure} 

Figs. \ref{fig:qvp_N_NCE} and \ref{fig:qvp_N_CE} illustrate  the QVP ‌versus the total number of packets, $N=N_\mathrm{roi}+N_\mathrm{bg}$, for  two eavesdropping scenarios, i.e.,  NCE and CE, respectively.   
In these figures,  two different locations are considered for the legitimate destination, which are shown by case I and case II in the figures. In Fig. \ref{fig:qvp_N_NCE}, cases I and II correspond, respectively, to the case of $\mathcal{D}$ being at  $[1.8,-0.8]$ and $[2,-1]$.   
In Fig. \ref{fig:qvp_N_CE}, $\mathcal{D}$ is located at  $[2.4,1.4]$ and $[2.6,1.6]$ for cases I and  II, respectively.   
 In figures \ref{fig:qvp_N_NCE} and \ref{fig:qvp_N_CE}, the resultant QVP  from our proposed DNN is compared with two related baselines.   
1) The maximum power (MP) transmission scenario, in which the transmit SNR of source, for  both types of confidential and public transmission,  is set to its maximum level $\gamma_\mathrm{max}$;  and, 
2) the equal power (EP) allocation between AN and information signal, in which we set $\zeta=0.5$, while the remainder of the transmission parameters are set to their optimum values obtained from our DNN.    
The performance gain of our learning-based approach compared with the mentioned baselines can be clearly observed from the figures. 
To be more specific, our proposed learning-based approach has come up with an intelligent and flexible procedure to take optimal transmission parameters according to different system configurations. This was discussed in details in Section V.  
 On the contrary, adopting MP or EP policies cannot provide the optimum  QVP since they blindly choose  the transmission parameters.  
 For instance, according to the MP scenario, 
 one can conclude that  sending an image with maximum available power does not necessarily yield to the minimum achievable QVP.  
 Instead,  we should take the image content and the general configurations of the wireless system into account to choose the best option for the transmit power.  
As demonstrated by Figs.  \ref{fig:qvp_N_NCE} and \ref{fig:qvp_N_CE}, the QVP of EP scenario is higher than our proposed scheme and the MP scenario. 
This observation highlights the importance of fine-tuning the AN signal based on different configurations \cite{IoTJ}.  
It can be inferred from these figures that if the power allocation ration $\zeta$ between the AN and information signal is not properly adjusted, the resultant QVP tends to large values although maintaining optimal transmit powers. 
The rational behind this effect is 
that the value one chooses for  $\zeta$ not only determines the amount of AN employed for confusing Eves, but also affects the remaining power budget for transmitting the information signal. 
Finally, 
 we can see that when the destination moves towards  $\mathcal{S}$, better QVPs can be achieved which is in accordance with one's intuition.

\begin{figure}\centering
	\begin{minipage}{0.45\textwidth}
		\centering
	\includegraphics
	[width=3.2in,height=2.4in,
	trim={0.1in 0.1in 0 0.1in},clip]
	{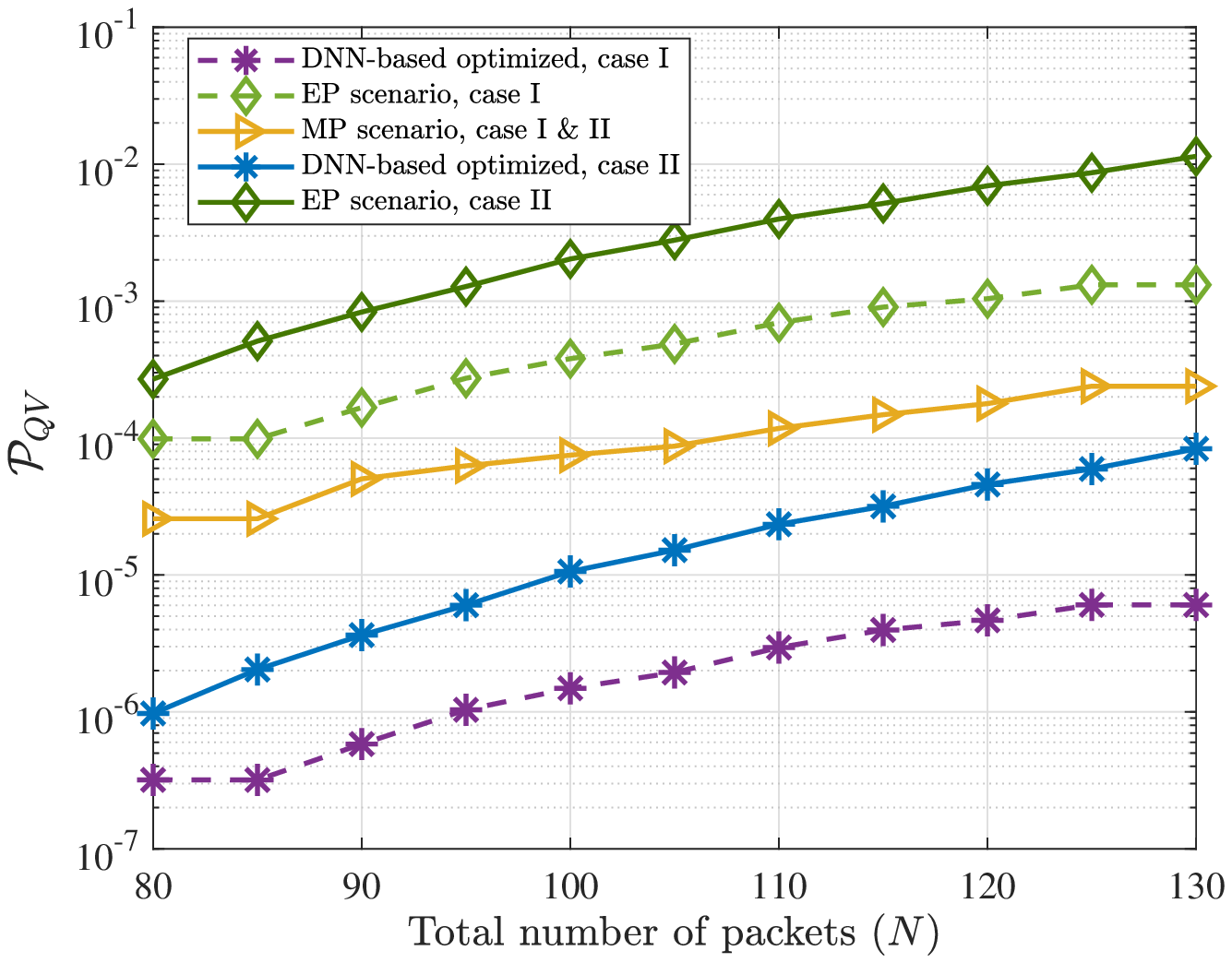}
	\caption{\small QVP vs. the total number of packets for the CE scenario with $\rho = 0.75$.}
	\label{fig:qvp_N_CE}\vspace{-7mm}
	\end{minipage}\hfill
	\begin{minipage}{0.45\textwidth}
	\centering
			\includegraphics
			[width=3.2in,height=2.4in,
			trim={0.1in 0.1in 0 0.1in},clip]
			{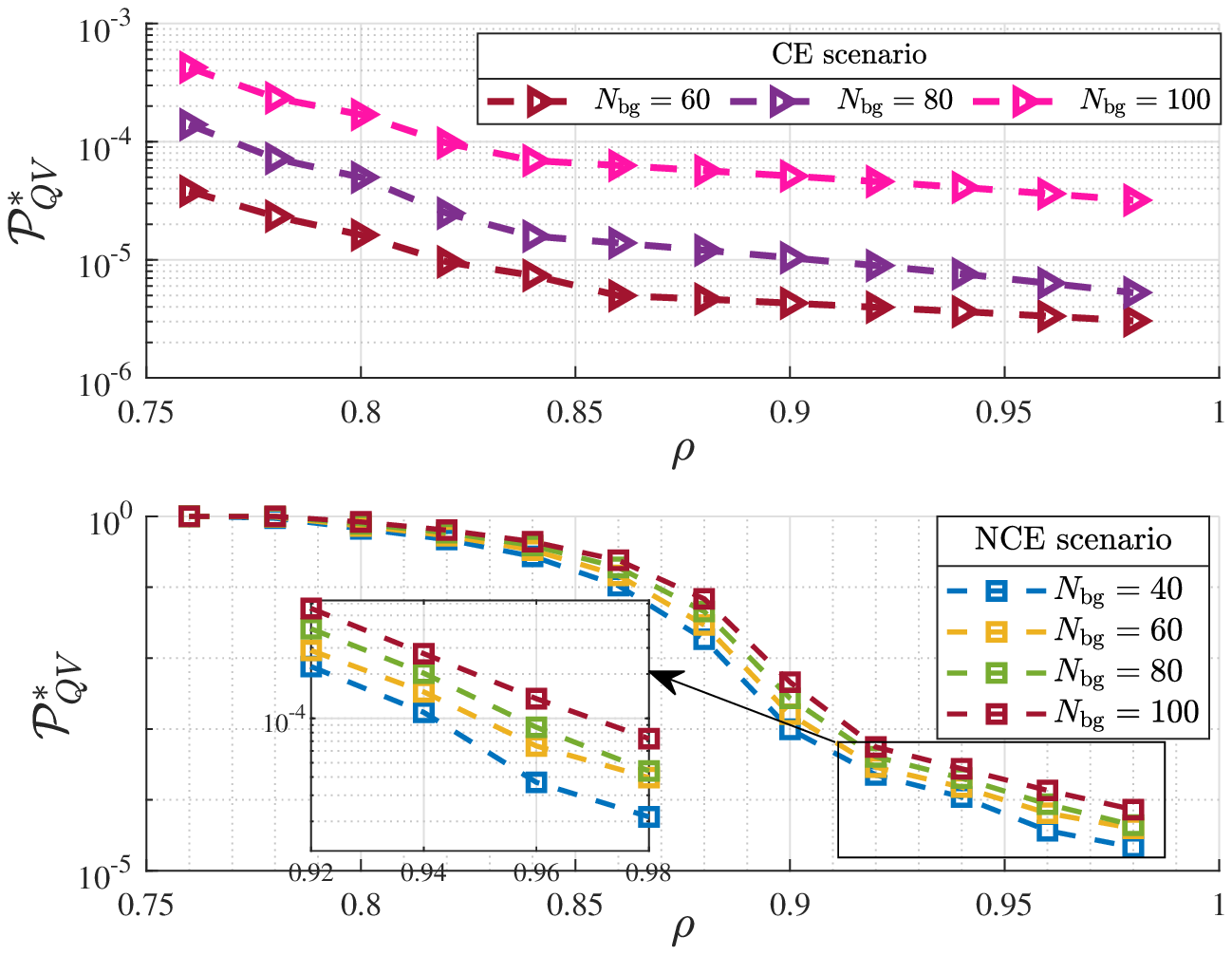}
			\caption{\small Optimized QVP versus the channel correlation coefficient $\rho$.}
			\label{fig:qvp_rho}\vspace{-7mm}
	\end{minipage}
\end{figure}

Fig. \ref{fig:qvp_rho} demonstrates our learning-based optimized QVP ($\mathcal{P}^\ast_{QV})$ ‌versus 
 the correlation coefficient between the legitimate channel,   $\boldsymbol{h}_{sd}$,  and the estimated one,  $\hat{\boldsymbol{h}}_{sd}$, for different values of $N_\mathrm{bg}$.  
In this figure, we consider the destination to be located at $[1.5,1.5]$ and $[3,3]$ for the CE and NCE scenarios, respectively.  
 From the figure, one can see the effect of imperfect channel estimation on the overall performance of the system.  
 Accordingly, by having accurate estimations of the legitimate channel,  i.e., higher values of $\rho$, better system performance is achieved in terms of the QVP. 
 Mathematically speaking, the higher the quality of the channel estimation is, the higher received SINR is obtained at $\mathcal{D}$, which can be easily verified from \eqref{gamma_D}.
 The reason is twofold. 
 First, with the increase in $\rho$,  more accurate beamforming vector $\boldsymbol{w}$ for transmitting the information-bearing signal can be designed at PHY based on \eqref{PHY-design}. Hence, larger SINRs are obtained at $\mathcal{D}$;   
 second, when the available estimation of the legitimate link becomes more accurate, the amount of AN signal which unwantedly exists in the legitimate link due to the  imperfect design of the AN beamforming,  $\mathbf{G}\boldsymbol{v}_\mathrm{AN}$, decreases. This can be inferred from \eqref{PHY-design} and \eqref{y_D}.    
 We also emphasize that having greater values for the SINR at destination also results in  higher rates for packet  transmission according to \eqref{L}.    
 With a similar discussion on the derived formulas,  
 having poor estimations of the legitimate channel may result in QVPs which might not be desirably small.
In Fig. \ref{fig:qvp_rho}, we also examine the optimized QVP for different values of $N_\mathrm{bg}$. Accordingly, increasing the number of $N_\mathrm{bg}$ will require more TSs to complete the file delivery which increases $\mathcal{P}^\ast_{QV}$  for a fixed $D_{lim}$. This can also be verified from \eqref{TD}, \eqref{N_tilde_bg}, and \eqref{qvp-2}.  
We note that according to this figure, one can see that a slight improvement in the quality of channel estimation leads to the fact that more packets can be sent within a given delay limit, while having the same QVP. 
For instance, if we set $\mathcal{P}^\ast_{QV} = 10^{-4}$, about $8\%$ increase in $\rho$ results in $25\%$ increase in the number of packets allowed to be sent within the delay limit $D_{lim}=30$ for the CE case. 
Similarly, for the NCE case,  by $3\%$ improvement in  $\rho$,  $1.5$ times more packets can be transmitted.

\section{Conclusions} 
In this paper, a learning-assisted content-aware  image transmission scheme is investigated,  over a  MISO channel 
in the presence of randomly distributed passive eavesdroppers.   
We consider the fact that not all regions of an image have the same priority from the security perspective. 
Thus, a hybrid multi-packet  transmission scheme is proposed 
to assure security while taking into account the delay limits of a practical image transmission scheme. 
To further reinforce the system's security,  the source packets are encoded by an FC-based method at the application layer. 
 A closed-form expressions for the QVP metric is derived to characterize the delay-aware performance  
 of our proposed scheme. 
Moreover, we  take advantage of a fully-connected DNN to minimize the QVP by maintaining optimized transmission parameters. 
Simulation results  illustrate that our proposed  learning-assisted scheme outperforms the state-of-the-art benchmarks by achieving  considerable gains in terms of security and  delay.

\appendices
\section{}
\vspace{0mm}
We first define $u = |\boldsymbol{w}^\mathsf{T}\boldsymbol{h}_{se_i}|^2$  and $v = ||\boldsymbol{h}_{se_i}^\mathsf{T}\mathbf{G}||^2$.
Notably, we have  $u \sim \text{Exp}(1)$, and $v \sim \mathcal{G} (n_\mathsf{T}-1,1)$. 
In addition, $u$ and $v$ are independent of each other, since $\boldsymbol{w}$ and $\boldsymbol{G}$ are orthogonal. 
Invoking \eqref{gamma_Ei} and by defining  $\xi\treq\frac{1/\zeta-1}{n_\mathsf{T}-1}$,  we first derive the CDF $F_{\gamma_{E_i}}(\omega)$ of $\gamma_{E_i}$ which can be expressed as
\begin{align}
	F_{\gamma_{E_i}}(\omega)=
	\mathsf{Pr}
	\left(u<\frac{\xi \zeta P_s v + r^\eta_i\sigma_n}{\zeta P_s}\omega\right) 
	&= 1-\E_v\left[e^{-\xi v \omega} e^{-\frac{r_i^\eta\omega}{\zeta P_s/\sigma_n}}\right] 
	\nonumber \\
	&
	\stackrel{(a)}{=}	1-\exp\left({-\frac{r_i^\eta\omega}{\zeta P_s/\sigma_n}}\right)(1+\xi \omega)^{1-n_\mathsf{T}}, 
\end{align} 
where ($a$) follows from  {\cite[Eq. (3.326.1)] {integ}}. Consequently, 
$F_{\gamma^{NCE}_{E}}(\omega)$  is formulated as follows 
\begin{align}\label{pre-fin}
	F_{\gamma^\mathsf{NCE}_{E}}(\omega)&=\E_{\Phi_E}\left[
	\prod_{i \in \Phi_E} \mathsf{Pr}(\gamma_{E_i}<\omega)
	\right] 
	\stackrel{(b)}{=}\exp\left(
	\frac{-2\pi\lambda_E}{(1+\xi\omega)^{n_\mathsf{T}-1}}\int_{0}^{\infty} r e^
	{\frac{-r^{\eta}\omega}{\zeta P_s/\sigma_n}}dr
	\right),
\end{align}
where ($b$) holds for the probability generating functional
lemma (PGFL) over PPP \cite{Stoch-Geom}. Finally, calculating \eqref{pre-fin} with formula
{\cite[Eq. (3.326.2)] {integ}} completes the proof.

\section{}
\vspace{0mm}
For the CE case, we first derive the Laplace transform of $\gamma_E$ as follows
\begin{align}
	\mathcal{L}_{\gamma^\mathsf{CE}_E}(s)&=
	\E_{\Phi_E,u,v}\left[e^{-s\sum_{i\in \Phi_E}^{}\gamma^\mathsf{s}_{E_i}}\right]\stackrel{(a)}{=}
	\E_{\Phi_E}\left[
	\prod_{i \in \Phi_E} \E_{u,v}\left[e^{-s\gamma^\mathsf{s}_{E_i}}\right]
	\right]
	\nonumber\\
	& 
	\stackrel{(b)}{=}
	\exp\left(\hspace{-1mm}-2\pi\lambda_E
	\underset{\mathcal{I}}{\underbrace{
			\E_{u,v}\Bigg[\int_{0}^{\infty}
			\hspace{-2mm}
			r_i\left(1-\exp({-s\gamma^\mathsf{s}_{E_i}})\right)dr_i \Bigg]}	
	}
	\right)\hspace{-1mm},  \label{int_dr}
\end{align}
where ($a$) follows from the fact that $\gamma_{E_i}$'s, $i\in\Phi_E$ are independent of each other over $u$ and $v$. In addition, ($b$)  holds for the PGFL over PPP. 
According to \eqref{gamma_Ei}, we can rewrite $\gamma^\mathsf{s}_{E_i}=\frac{a_1{u}}
{a_2v+r^\eta_i}$, with $a_1=\zeta \frac{P_s}{\sigma_n}$ and $a_2=(1-\zeta)\frac{P_s}{\sigma_n}/(n_\mathsf{T}-1)$. 
Consequently, we proceed to solve the 
expected-valued integral term in \eqref{int_dr}. Hence, we have
\begin{align}
	\mathcal{I}&\stackrel{(c)}{=} 
	\int_{0}^{\infty} r_i  \hspace{1mm} \E_{v}\left[
	\frac{a_1s}{a_1s+a_2v+r^\eta_i}	\right] dr_i 
	\stackrel{(d)}{=} B\left(\frac{2}{\eta},1-\frac{2}{\eta}\right) \frac{a_1s}{\eta}	\int_{0}^{\infty} (a_1s+a_2v)^{2/\eta-1} f_v(v)dv
	\nonumber\\
	&\stackrel{(e)}{=}
	\mathcal{B} e^{\frac{sa_1}{2a_2}} (a_1s)^{\frac{n_\mathsf{T}-1+2/\eta}{2}}
	W_{\frac{1-n_\mathsf{T}+2/\eta}{2},\frac{2-n_\mathsf{T}-2/\eta}{2}}\left(\frac{a_1s}{a_2}\right) 
\end{align} 
where  ($c$) follows from the fact that $u$ and $v$ are independent,  ($d$) results from the transformation $r^\eta_i=y$ and using {\cite[Eq. (3.194.3)] {integ}}, and ($e$) follows by changing the variable $a_1s+a_2v=z$ and then using {\cite[Eq. (3.383.4)] {integ}}.  Moreover, $f_v(v)=\frac{v^{n_\mathsf{T}-2}e^{-v}}{\Gamma(n_\mathsf{T}-1)}$ and 
$\mathcal{B}=  \frac{B\left(\frac{2}{\eta},1-\frac{2}{\eta}\right)}{\eta}(a_2)^{\frac{1-n_\mathsf{T}+2/\eta}
	{2}}$.

\section{}
After deriving $\mathcal{L}_{\gamma^\mathsf{CE}_E}(s)$, the CDF of $\gamma^\mathsf{CE}_E$ can be derived via utilizing inverse Laplace transform. However, this method requires considerable computation complexity, hence is almost  intractable. 
Therefore, we utilize the following steps to provide an approximation of $	\overline{F}_{\gamma^\mathsf{CE}_{E}}(\omega)$ by introducing an intermediate RV denoted by $\mathfrak{I}$ which obeys normalized Gamma distribution with parameter $K$. Therefore, we have   
\begin{align}
	\overline{F}_{\gamma^\mathsf{CE}_{E}}(\omega)&=\E_{\Phi_E}
	\left[\mathsf{Pr}\left(	{\gamma^\mathsf{CE}_{E}}>\omega
	\right)\right]
	\stackrel{(a)}{=}
	\E_{\Phi_E}\left[\mathsf{Pr}\left(
	\frac{\gamma^\mathsf{CE}_{E}}{\omega}
	>\mathfrak{I}
	\right)\right]
	\stackrel{(b)}{\stackrel{>}{\approx}}
	\E_{\Phi_E}\left[
	\left(
	1-\exp\left(-\varphi	\frac{\gamma^\mathsf{CE}_{E}}{\omega}\right)
	\right)^K
	\right],		
\end{align}
where ($a$) holds due to the fact that a normalized Gamma RV converges to identity when its shape parameter goes
to infinity, and ($b$) follows the CDF bound of a normalized Gamma RV \cite{norm-Gamma}. 
Finally, through using binomial expansion, we obtain the tight approximation given in \eqref{ccdf}, and the proof is completed.

\section{}
To further simplify \eqref{TD}, 
recall that the proposed protocol  terminates image transmission as soon as $N_\mathrm{bg}+N_\mathrm{roi}$ fountain packets 
are correctly received at $\mathcal{D}$.  
Additionally, owing to the proposed semi-adaptive multi-packet transmission of public BG  packets based on \eqref{L}, 
which ensures all BG ‌packets to be successfully delivered at $\mathcal{D}$, {public transmission} of BG packets  is not faced with outage event.  
Consequently, if we denote the outage probability for the legitimate link by $\Omega$, we can  
model $N_D^\mathrm{out}$ as an RV following the negative binomial distribution $\mathcal{NB}(
\frac{N_\mathrm{roi}}{L_s}
; \Omega)$, with the  PMF given by
\begin{align}\label{pmf_leg}
	{f}_{N_D^\mathrm{out}}(k)&=\mathsf{Pr}\left(
	N_D^\mathrm{out}=k\right)
	=\binom{\frac{N_\mathrm{roi}}{L_{s}}+k-1}{k}\Omega^k (1-\Omega)^{\frac{N_\mathrm{roi}}{L_{s}}},   
	k\geq0, 
\end{align}   
where  $\Omega$ can be derived as 
\begin{align}\label{Omega}
	\Omega&=\mathsf{Pr}\bigg(
	\frac{BT}{b}
	\log_2(1+\gamma^\mathsf{s}_D)<L_s\Big\vert\Psi_1
	\bigg)
	\stackrel{(a)}{=}
	\frac{\mathbf{1}_{(\theta\geq\kappa_s\nu)}}	{\mathsf{Pr}(\Psi_1)}
	\Big(\gamma(n_\mathsf{T},\frac{\theta}{\kappa_s})-\gamma(n_\mathsf{T},\nu)\Big)\Big/{\Gamma(n_\mathsf{T})}, 
\end{align} 
with ($a$) following from {\cite[Eq. (3.382.5)] {integ}}.   
Therefore, based on \eqref{TD}, \eqref{pmf_leg}, and \eqref{Omega},    
by substituting $N_D^\mathrm{out} = k-\tilde{N}$ in \eqref{pmf_leg},  \eqref{pmf_TD} is obtained.

\section{}
By assuming the fact that $T_E=\frac{N_\mathrm{roi}}{L_{s}}+N_E^\mathrm{out}$ and similar to the procedure of deriving \eqref{pmf_TD}, one can deduce that $T_E$ has the following PMF, denoted by  $f_{T_E}(k)$. 
\begin{align} 
	\hspace{-2mm}
	f_{T_E}(k)&=
	\binom{k-1}{k-\frac{{N}_\mathrm{roi}}{L_{s}}}\Lambda_{c}^{k-\frac{{N}_\mathrm{roi}}{L_{s}}}(1-\Lambda_c)^{\frac{{N}_\mathrm{roi}}{L_{s}}},
	k\geq \frac{{N}_\mathrm{roi}}{L_{s}},  
\end{align}
where  $c \in \{\mathsf{NCE},\mathsf{CE}\}$.    
{We remark that for Eves, we have assumed that the accumulation of  $N_\mathrm{roi}$ confidential packets will suffice to intercept the image; hence, $N^{\mathrm{out}}_E$ represents the number of TSs that the public packets have been sent by $\mathcal{S}$, or the confidential packets have not been correctly recovered at wiretappers.} 
Accordingly, for the NCE case we have     
\begin{align}\label{Lambda_NCE}
	\Lambda_\mathsf{NCE} &=  
	{\mathsf{Pr}(\Psi_1)}
	\mathsf{Pr}\bigg\{BT\log_2(1+\gamma^\mathsf{NCE}_E)<R	{\Big\vert\Psi_1}\bigg\}+{\mathsf{Pr}(\Psi_0)}\
	\nonumber\\
	&\stackrel{(a)}{=}
	{\mathsf{Pr}(\Psi_1)} \exp\left(-\beta \lambda_E (\zeta P_s/\theta)^{\frac{2}{\eta}} \bigg(1+\frac{1/\zeta-1}{n_\mathsf{T}-1}\theta\bigg)^{1-n_\mathsf{T}}
	\right)
	+{\mathsf{Pr}(\Psi_0)},  	
\end{align} 
where ($a$) follows from Lemma \ref{lemma1}. 
Similarly, for the CE scenario, we can rewrite  
\begin{align}\label{Lambda_CE}
	\Lambda_\mathsf{CE} &=  
	{\mathsf{Pr}(\Psi_1)}\mathsf{Pr}\bigg\{BT\log_2(1+\gamma^\mathsf{CE}_E)<R\Big\vert\Psi_1\bigg\}+{\mathsf{Pr}(\Psi_0)} 
	\nonumber\\
	&
	\stackrel{(b)}{=}
	{\mathsf{Pr}(\Psi_1)}
	\left[1-\sum_{k=0}^{K}\binom{K}{k}(-1)^k
	\mathcal{L}_{\gamma^\mathsf{CE}_E}\left(\frac{k\varphi}{\theta}\right)\right]  +{\mathsf{Pr}(\Psi_0)}, 	
\end{align} 
with ($b$) following from Lemma \ref{lemma2}, and $\mathcal{L}_{\gamma^\mathsf{CE}_E}(\cdot)$ derived in \eqref{L(s)}.


\end{document}